%% file: main.tex
\documentclass{article} %
\PassOptionsToPackage{table}{xcolor}
\usepackage{iclr2025_conference}

\usepackage{url}
\usepackage{xspace}
\usepackage{amssymb}
\usepackage{xcolor}          %
\usepackage{algorithm}
\usepackage{algpseudocode}
\usepackage{multirow}
\usepackage{multicol}
\usepackage{float}
\usepackage{booktabs}
\usepackage{amsmath,amsthm} 
\usepackage{hyperref}
\IfFileExists{accessibility.sty}{%
  \usepackage{accessibility}%
}{%
  \PackageWarningNoLine{IntHQ}{accessibility.sty unavailable; continuing without tagged PDF support}%
}
\newtheorem{theorem}{Theorem}
\newtheorem{lemma}{Lemma}
\newcommand{\best}[1]{\cellcolor{dreamblue!15}#1}     
\newcommand{\second}[1]{\textcolor{dreamblue}{#1}}    
\newtheorem{proposition}{Proposition}

\hypersetup{
  colorlinks=true,
  linkcolor=dreamblue,
  citecolor=dreamblue,
  urlcolor=dreamblue,
  filecolor=dreamblue,
}

\renewcommand{\dreamxrunhead}{DreamX}                 %
\renewcommand{\dreamxvenue}{DreamX}  %
\renewcommand{\dreamxdate}{August 2026}                     %

\newcommand{\modelname}{\textbf{IntHQ}\xspace}

\title{IntHQ: Task-Interactive Hierarchical Query on Dual-Stream Representations for Generative Recommendation}

\author{%
Junjie Sun, Longfei Xu\thanks{Project lead.}, Huimin Yan, Wei Luo, Kaikui Liu, Xiangxiang Chu
\AND
{\normalfont DreamX, Alibaba Group}%
}

\iclrfinalcopy %

\begin{document}

\maketitle

\input{tex/abstract}

\input{tex/intro}

\input{tex/pre2}

\input{tex/Methodology}

\input{tex/exp}

\input{tex/related_work}

\input{tex/conclusion}

\bibliography{iclr2025_conference}
\bibliographystyle{iclr2025_conference}

\appendix

\input{tex/appendix}

\end{document}

%% file: tex/abstract.tex
\begin{abstract}

Multi-task learning over heterogeneous data is fundamental to modern recommendation, while generative models are emerging as the backbone of next-generation recommenders. However, the integration of multi-task learning into the generative paradigm remains largely unexplored.

Existing multi-task recommenders, in both discriminative and generative paradigms, extract task-relevant features from a single task-agnostic representation and wire tasks into a predefined conversion funnel. We show that this scheme is inherently prone to a threefold collapse. \textbf{Source collapse}, where task-specific signals are injected late and diluted in the shared latent space. \textbf{Relational collapse}, where task dependencies are either implicitly absorbed by the backbone or statically fixed by predefined funnels. \textbf{Hierarchical collapse}, where tasks depend on features at different scales and shift across training stages.

We propose \modelname, a multi-task generative recommender with three components, each alleviating one collapse. \textbf{Dual-Stream Decoupling (DSD)} injects task identity into computation stream early and separates the shared context stream from the task-specific stream, alleviating signal dilution. \textbf{Task-Interactive Modeling (TIM)} replaces the predefined funnel with explicit cross-task interaction, letting each task condition on the realized outcomes of its predecessors with learned, input-adaptive strength.  \textbf{Hierarchical Querying (HQ)} lets each task gather multi-scale information across different layers at different training stages.

In offline evaluations, \modelname consistently outperforms competitive encoder backbones under four representative task-head configurations. Deployed in production on Amap, serving hundreds of millions of users for travel recommendation, \modelname yields a 1.60\% relative UVCTR lift. The implementation of \modelname is available at \href{https://github.com/AMAP-ML/DreamX-Rec/}{\texttt{github.com/AMAP-ML/DreamX-Rec}}.

\end{abstract}

%% file: tex/intro.tex
\section{Introduction}
Modern industrial recommender systems increasingly rely on multi-task learning (MTL) to jointly optimize multiple business objectives~\citep{caruana1997multitask,wang2023multi,zhang2025advances}. By sharing information across tasks, MTL reduces serving cost and improves robustness under bias, sparsity, and cold-start settings~\citep{zhang2025advances}. Existing work mainly follows three directions. The first focuses on allocating capacity across tasks, evolving from shared-bottom~\citep{caruana1997multitask} to MMoE~\citep{ma2018modeling} and PLE~\citep{tang2020progressive} to mitigate negative transfer. The second investigates the relation among tasks, such as ESMM~\citep{ma2018entire} and AITM~\citep{xi2021modeling}, which leverage dependencies among user actions to reduce sparsity and selection bias. The third operates at the optimization level. GradNorm~\citep{chen2018gradnorm} rebalances per-task loss scales, PCGrad~\citep{yu2020gradient} projects away the conflicting component between task gradients. CAGrad~\citep{liu2021conflict} and Nash-MTL~\citep{navon2022multi} solve for an update direction that every task accepts. These methods~\citep{tang2020progressive,sheng2021one,wang2025home,ma2018modeling,ma2018entire} are developed for discriminative, feature-engineering-heavy ranking stages within a multi-stage cascade, and do not transfer naturally to a unified sequence formulation.

\begin{figure}[ht]
  \centering
 \includegraphics[width=\linewidth]{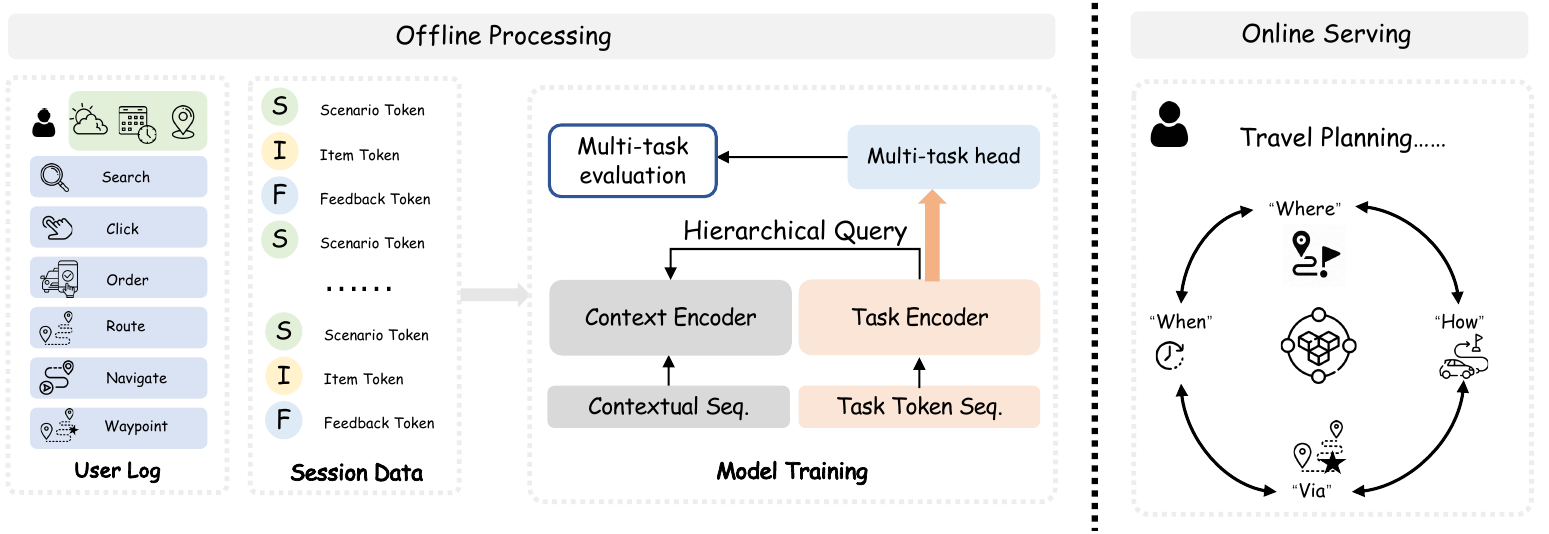}
    \caption{Overview of \modelname. Left: offline processing turns user logs into session-structured data with scenario/item/feedback (SIF) tokens. The context stream and the task-token stream are decoupled. Task tokens interact within task encoder and hierarchically query the context encoder. Right: online serving jointly produces the \emph{when}, \emph{where}, \emph{how}, and \emph{via} decisions for travel recommendation, varies with time and scenarios.}
  \label{fig:intro_ovewview}
\end{figure}

Inspired by Large Language Models (LLMs), the community has recently moved toward Generative Recommendation (GR)~\citep{hou2026survey,zhai2024actions,deng2025onerec,rajput2023recommender,zhang2026onetrans}. Tailored for high-cardinality, non-stationary streaming data~\citep{zhai2024actions}, GR leverages Transformer-based sequence models to reformulate recommendation as a direct sequence generation task. By unifying historical behaviors, feedback, context, and candidates into a single tokenized sequence, GR collapses the fragmented cascade into one backbone and enables end-to-end model scaling~\citep{zhai2024actions,yan2026inttravel,yan2025intsr,shi2024unisar,shi2026mining}.

Although GR is becoming the backbone of a new generation of recommender systems, how to meet the multi-objective demands inherent to industrial settings remains an open problem. GR is trained under a single-objective, next-item prediction paradigm ~\citep{rajput2023recommender,zhai2024actions,deng2025onerec} and does not natively support the joint optimization of multiple heterogeneous objectives. To bridge this mismatch, existing multi-task generative recommenders fuse task modeling into the shared sequence representation, either conditioning tasks at output-side heads~\citep{yan2026inttravel} or injecting task tokens that are encoded jointly with behavior tokens~\citep{tang2026onerank,liang2026hgenpush}. Either way, task signals become entangled with the shared representation. Multi-task learning with the generative paradigm still faces three collapses.

\begin{itemize}
\item \textbf{Source collapse.} Task-specific signals are injected late and diluted in the shared space. Context tokens and task tokens are functionally heterogeneous. Encoded as a continuous stream, context tokens must characterize user behavior across diverse scenarios, yielding a high-dimensional representation coupled with user intention. Task tokens, in contrast, are grounded on task-specific semantic manifolds~\citep{bengio2013representation,vafidis2025disentangling}. Task tokens require distinct, task-oriented emphases.

\item \textbf{Relational collapse.} Task dependencies are either implicitly absorbed by the backbone or statically fixed by predefined funnels~\citep{ma2018entire,xi2021modeling}. Interactions entangled inside the shared stream can only be curbed by cross-task gradient isolation that discards useful cross-task signal~\citep{tang2026onerank}. Yet some multi-task dependencies are complex and cannot be specified a priori~\citep{yan2026inttravel}. 

\item \textbf{Hierarchical collapse.} Different layers of a deep backbone encode context at different granularities. From fine-grained local patterns to abstract long-range intent, heterogeneous tasks rely on different depths at different training stages~\citep{sun2020adashare}. Taking only the last layer or a fixed per-layer scalar gate~\citep{yan2026inttravel} flattens this hierarchy into one representation shared by all tasks.
\end{itemize}

The three collapses above take a concrete form in map and travel services. As a Location-Based Service (LBS) provider, we resolve four coupled decisions within a single session, namely a user's departure time, destination, travel mode, and waypoint, which we refer to as \emph{when}, \emph{where}, \emph{how}, and \emph{via}. Figure~\ref{fig:intro_ovewview} illustrates the deployed service.

Within a single session, the four tasks draw on different evidence. \emph{Where} and \emph{when} are settled by a coarse periodicity of the long-term history. \emph{Via} turns on one transient need inside the current session, such as refueling. \emph{How} responds to the real-time condition of the moment. A single task-agnostic embedding keeps the dominant statistic and drops the minority evidence, therefore creates a bottleneck. Nothing placed downstream can recover the discarded information. \textbf{Task signals therefore have to be injected early.} Meanwhile, the context computation compress sequences into what persists, whereas task computation is discriminative in current situation. The two roles have different optima for the same parameters. A shared block can only settle between them. \textbf{The task and context computation therefore have to live in a decoupled parameter space.} Moreover, the tasks induced by user behaviors are tightly coupled rather than independent across scenarios, as shown in bottom right of Figure~\ref{fig:intro_ovewview}. The destination in commuting scenario shapes travel time and mode, but an uncertain destination on a trip loosens this dependence, so no fixed order holds across scenarios. \textbf{ Dependencies that shift with real-world constraints call for an adaptive task-interaction framework.} Meanwhile, \textbf{the tasks differ substantially in vocabulary size and in the sparsity of informative behavioral signals}, so it is beneficial to route representations at different scales adaptively. 

Motivated by these observations, we propose Task-\textbf{Int}eractive \textbf{H}ierarchical \textbf{Q}uery (\modelname). The main contributions of this work are summarized as follows:

\begin{itemize}
    \item \textbf{Novel Multi-task Generative Recommendation Paradigm.} We identify three collapses in multi-task generative recommendations and propose \modelname with three components that address the three collapses respectively. \textbf{Dual-Stream Decoupling.} We attach four constructed task tokens at each position of the behavior sequence to form an independent task-token sequence, with separated parameters decoupled from the context computation stream. \textbf{Task-Interactive Modeling.} Within the task-token stream, we introduce an attention-based, explicit, and adaptive interaction among tasks. \textbf{Hierarchical Querying.} Task tokens act as queries while the context-stream representations serve as keys and values, aggregating informative signals across different layers.

    \item \textbf{Comprehensive Experimental Validation.} We conduct extensive experiments on a real-world industrial dataset IntTravel~\citep{yan2026inttravel}. \modelname provides highly informative task signals that consistently benefit all task heads, and outperforms the latest baselines. Ablation study also validates the contribution of each component.

    \item \textbf{Deployment and evaluation.} We deploy \modelname in production at Amap. \modelname currently serves hundreds of millions of users at tens of thousands of Query Per Second (QPS) over tens of millions of POIs, and powers key products including app launch map, terminal recommendation, travel-mode recommendation, and along-the-way recommendation. Extensive online A/B tests show consistent gains, delivering a 1.60\% relative UVCTR (CTR on unique user views) lift while keeping average inference latency within 40 milliseconds.
\end{itemize}

%% file: tex/pre2.tex
\section{Preliminaries}
In this section, we formalize multi-task learning in travel recommendation. Building on this formulation, we theoretically and empirically expose a threefold collapse inherent in conventional architectures, which inspires the design principles of our approach.

\subsection{Problem Formulation}
Given all user interactions set $\mathcal{A}$, consists of the context sequence summarizing the user profile, spatio-temporal state, and history, the system jointly resolves a travel intention that decomposes into $K$ coupled sub-decisions indexed by $k\in\mathcal{K}$. We instantiate $\mathcal{K}=\{$\emph{when}, \emph{where}, \emph{how}, \emph{via}$\}$. The ground-truth labels are defined as $\mathbf{y}=[\,y_k \mid k\in\mathcal{K}\,]$, each $y_k$ valued in a candidate set $\mathcal{I}_k$ (possibly huge, e.g.\ POIs). Training pairs in dataset $\mathcal{D}$ are $(\mathcal{A},\mathbf{y})\sim\mathcal{D}$.

User behavior is logged by session. Interaction set $\mathcal{A}$ is an ordered sequence assembled
from three token types. 
\begin{equation}
\label{eq:session}
\mathcal{A}_u=\big(\,(S_t,\,I_t,\,F_t)\,\big)_{t=1}^{T_u},
\qquad S_t\in\mathcal{S},\; I_t\in\mathcal{I},\; F_t\in\mathcal{F},
\end{equation}

where each step $t$ contributes three context token types. Scenario token $S_t$ encodes spatio-temporal state of the $t$-th action. Item token $I_t$ denotes the POI interacted with at step $t$ and its attributes. Feedback token $F_t$ signals the user’s interaction type, such as clicking or searching. In addition, a set of user profile tokens $\mathcal{U}_u=(U_1,\dots,U_m)$, encoding long-term static attributes (e.g., demographics), is prepended to the sequence, i.e., $\mathcal{A}'_u=\big(\,\mathcal{U}_u,\,(S_t, I_t, F_t)_{t=1}^{T_u}\big) \in \mathcal{A}$, serving as a global context visible to all subsequent tokens. To decouple task tokens $q^{t}_{k} \in \mathcal{Q}$ and context tokens, we define task token sequence and attach labels to related tasks as:

\begin{equation}
\label{eq:task-stream}
\mathcal{Q}_u=\big(\,(q^{\,t}_{\text{when}},\,q^{\,t}_{\text{where}},\,q^{\,t}_{\text{how}},\,q^{\,t}_{\text{via}})\,\big)_{t=1}^{T_u},
\qquad q^{\,t}_{k}\in\mathcal{Q}_{k},\;\; k\in\mathcal{K},
\end{equation}

A model with parameters $\Theta$ encodes the request into a representation $z=\mathrm{Enc}_{\Theta}(\mathcal{A})$ and is trained by maximum likelihood of the joint target.
\begin{equation}
\label{eq:mle}
\hat{\Theta}=\arg\max_{\Theta}\sum_{(\mathcal{A},\mathbf{y})\in\mathcal{D}}
\log P\big([\,y_k\mid k\in\mathcal{K}\,]\mid \mathcal{A};\Theta\big),
\end{equation}

\subsection{Pilot Study}
\label{sec:pilot}

We score a model by the joint log loss of the outcome tuple, whose Bayes optimum is $R^\star=H(\mathbf y\mid\mathcal A)$, where $H$ denotes conditional entropy. Conventional multi-task
recommenders, in both the discriminative and the generative paradigm, share one pattern. They compress $\mathcal A$ into a single task-agnostic code $z=\mathrm{Enc}(\mathcal A)$ and score the tuple with heads that are conditionally independent given $z$, so the predicted law is $\prod_k q_k(y_k\mid z)$. We call this \emph{Task-Agnostic Encoding} (TAE), and its optimal risk is $R_{\mathrm{TAE}}(z)=\sum_k H(y_k\mid z)$. Theorem~\ref{thm:unified} splits the excess risk of TAE into two non-negative terms, which we name the \textbf{source} and the \textbf{relation} collapse. All proofs and analysis on \textbf{source} and \textbf{relation} collapse are in Appendix~\ref{app:proofs}. The \textbf{hierarchy} collapse is established empirically in \S\ref{sec:pilot-hier}.

\begin{theorem}[Collapse decomposition]
\label{thm:unified}
Assume every $\mathcal I_k$ is finite, $z=\mathrm{Enc}(\mathcal A)$ is a measurable deterministic map, and the heads range over all conditional laws with no parameter shared across them. Then
\begin{equation}
\label{eq:unified}
R_{\mathrm{TAE}}(z)-R^\star
=\underbrace{\sum_k I(\mathcal A;y_k\mid z)}_{\Delta_{\mathrm{src}}(z)}
+\underbrace{\mathrm{TC}(\mathbf y\mid\mathcal A)}_{\Delta_{\mathrm{rel}}},
\end{equation}
where $I(\mathcal A;y_k\mid z)=H(y_k\mid z)-H(y_k\mid\mathcal A)$ is the information about $y_k$ that the encoding discards, namely total
correlation ~\citep{watanabe1960information,ver2014discovering} $\mathrm{TC}(\mathbf y\mid\mathcal A)=\sum_k H(y_k\mid\mathcal A)-H(\mathbf y\mid\mathcal A)$.
Both terms are non-negative, so $R_{\mathrm{TAE}}(z)\ge R^\star+\Delta_{\mathrm{rel}}$ for every encoder, with equality if and only if the encoding $z$ marginally sufficient for every task.
\end{theorem}

The two terms have disjoint dependence and therefore demand different structures. $\Delta_{\mathrm{src}}(z)$ carries the encoder and vanishes once $z$ is marginally sufficient per task, so a better encoder can shrink it. $\Delta_{\mathrm{rel}}$ denotes the floor of relation collapse and carries no $z$, so it bounds every factorized read-out from below no matter how large the encoder is or how it is trained. Shrinking $\Delta_{\mathrm{src}}$ asks for task conditioning inside the encoder, which we realize as \emph{Dual-Stream Decoupling}. Alleviating $\Delta_{\mathrm{rel}}$ requires decisions that condition on one another, which we realize as \emph{Task-Interactive Modeling}.

%% file: tex/Methodology.tex
\section{Methodology}

\begin{figure*}[t]
\centering
\includegraphics[width=\linewidth]{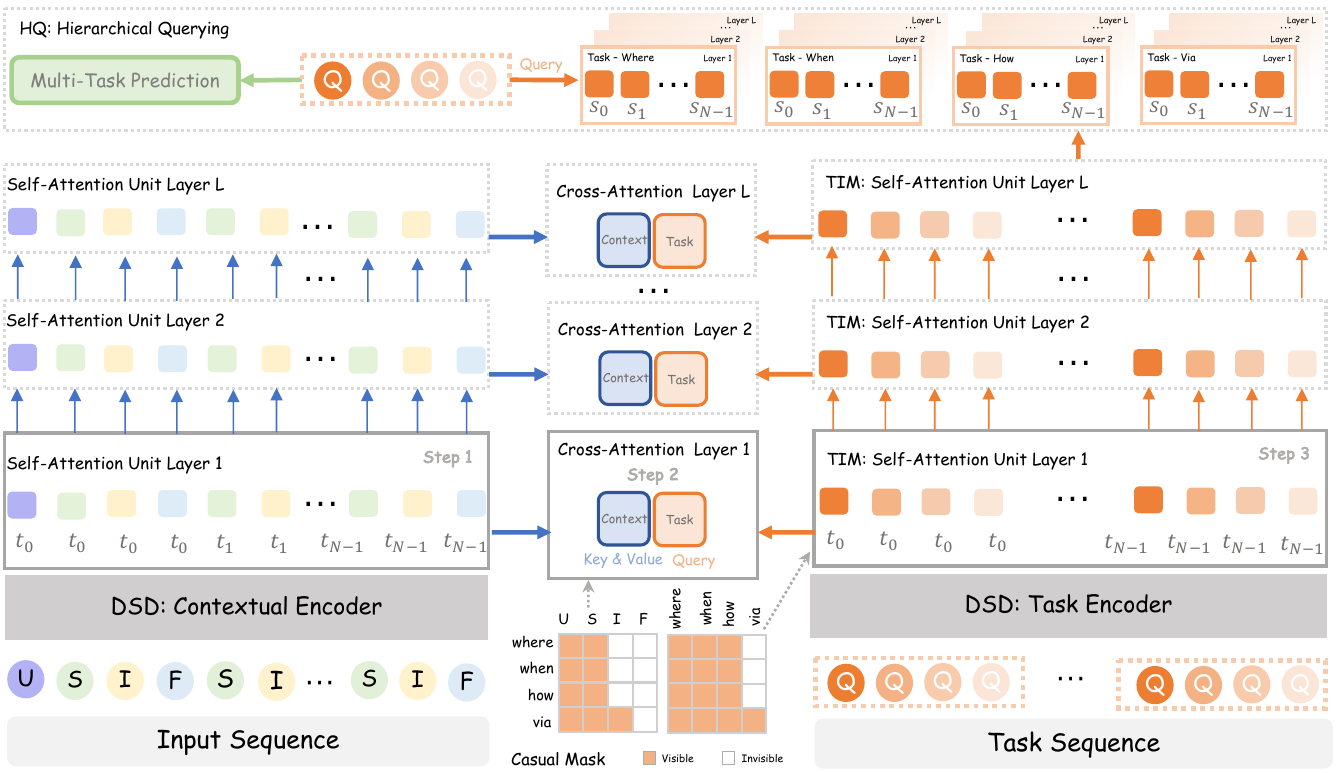}
\caption{\textbf{The detailed architecture of \modelname. \modelname is composed of three components, Dual-Stream Decoupling, Task-Interactive Modeling and Hierarchical Querying. The two streams are encoded separately and never concatenated. The mask between task tokens and context tokens is illustrated in the figure below, while visibility within the context tokens follows standard causal attention.}}
\label{fig:overview}
\end{figure*}

In this section, we present \modelname. As illustrated in Figure \ref{fig:intro_ovewview}, \modelname targets multi-task travel recommendation, where predictions are grounded in real-world spatiotemporal constraints and the dependencies among tasks shift dynamically with the scenario. \modelname targets three collapses found in the pilot study (\S\ref{sec:pilot}) with three core components as shown in Figure \ref{fig:overview}. \textbf{Dual-Stream Decoupling (DSD)}, targeting the source collapse, separates context encoding from task encoding. \textbf{Task-Interactive Modeling (TIM)} alleviates the relation collapse by enabling explicit and direct interactions among tasks. \textbf{Hierarchical Querying (HQ)} lets each task token gather information at multiple scales across layers. In the following subsections, we describe each component in detail.

\subsection{DSD: Dual-Stream Decoupling}
\label{sec:dsd}

The Dual-Stream Decoupling (DSD) module captures user intent and task-specific inclinations. DSD takes the session-structured user behavior sequence $\mathcal{A}$ as input and correspondingly constructed sequence of task tokens $\mathcal{Q}_u$. Two streams are encoded and interacted in decoupled latent spaces.

\subsubsection{Sequence Construction.} 

Raw interaction logs are heterogeneous. Each decision is made under an observable context, lands on a concrete item, and ends with a certain type of feedback. We map every logged session to three type context tokens. The scenario token $S$ encodes the decision context, such as time, location and weather. The item token $I$ encodes the consumed object and its attributes, e.g., POIs that users interacted. The feedback token $F$ encodes how the user engaged with the item, e.g., click, search, or order. A small set of profile tokens $U$ carrying long-term user attributes is appended once per sequence, serving as a global, time-invariant context.

The sequence $\mathcal{A}_u$ is organized session by session in temporal order, with all tokens of a session sharing the same timestamp, as in Eq.~\ref{eq:session}. The session is the natural unit of a decision. In travel, one trip jointly realizes \emph{when} to depart, \emph{where} to go, \emph{how} to travel, and \emph{via} which route. 

In parallel to $\mathcal{A}_u$, we construct the task stream $\mathcal{Q}_u$, as in Eq.~\ref{eq:task-stream}. Each task token $q_k$ is a task-specific learnable parameter shared across all sequences. It serves as a query template for task $k$ and is instantiated once per session. Task tokens carry the supervision of its task and query the context under a causal constraint that forbids label leakage.

\subsubsection{Decoupled Dual-Stream Attention.}
 
The two streams are encoded by independent attention cores with disjoint parameters. Within each layer, context core $\mathrm{Attn}_{\mathrm{ctx}}$ and task core $\mathrm{Attn}_{\mathrm{q}}$ encoding proceeds in two steps.
 
\paragraph{Context self-attention.}
The context stream applies self-attention over all context tokens.
\begin{equation}
\tilde{H}^{(\ell)}_{\mathrm{ctx}}
= \mathrm{Attn}_{\mathrm{ctx}}\!\big(
    H^{(\ell)}_{\mathrm{ctx}},\;
    H^{(\ell)}_{\mathrm{ctx}}
  \big),
\end{equation}
where $H^{(\ell)}_{\mathrm{ctx}}$ denotes the context representations at layer $\ell$. This step distills a task-agnostic summary of user intent from the behavioral history.

\paragraph{Task context cross-attention.}
Each task token queries the context stream.
\begin{equation}
\label{eq:tcc}
\tilde{H}^{(\ell)}_{\mathrm{q}}
= \mathrm{Attn}_{\mathrm{q}}\!\big(
    H^{(\ell)}_{\mathrm{q}},\;
    H^{(\ell)}_{\mathrm{ctx}}
  \big),
\end{equation}
so that the query, key, and value projections of $\mathrm{Attn}_{\mathrm{q}}$ are learned \emph{exclusively} from task-token gradients, while the key--value context is supplied by the context stream.  Crucially, the key--value set is drawn from the \emph{input} of the current layer $H^{(\ell)}_{\mathrm{ctx}}$, not from the self-attended output $\tilde{H}^{(\ell)}_{\mathrm{ctx}}$, so that the two steps remain
computationally independent within a layer.

\paragraph{Causal visibility.}
All attention operations in \modelname, context self-attention, task--context cross-attention, and the task self-attention introduced in \S\ref{sec:tim} obey the same visibility rule. Query at position $p_q$ may attend to a key at position $p_k$ if and only if $c(p_k) \le c(p_q)$, where $c(\cdot)$ is a scalar causal index assigned during sequence construction. The index encodes two constraints simultaneously: (i)~\emph{cross-session causality.} Tokens in earlier sessions are visible to all later tokens; and (ii)~\emph{intra-session ordering.} Within a session, a token sees only the facets that causally precede it, e.g., a \textit{via} token sees the realized \textit{where} token, but not vice versa. Profile tokens receive the smallest causal index and are therefore visible to every query. Following prior work, we add a learnable relative-position bias and a time-delta bias to the attention module, both indexed by the same positional structure that defines the causal index $c(\cdot)$. Figure~\ref{fig:overview} visualizes the causal rule in task-context cross attention and self-attention in task tokens, while visibility within the context tokens follows standard causal attention.

\subsection{TIM: Task-Interactive Modeling}
\label{sec:tim}

Dual-Stream Decoupling equips each task token with a context-conditioned representation via cross-attention, but these representations are computed independently. Task tokens do not yet see one another. In real-world multi-task recommendation, the dependency of tasks is usually non-trivial. Take travel recommendation as an example, the choice of destination constrains the feasible transport modes, which in turn narrows the departure time. Leaving such dependencies to be absorbed implicitly by the shared context stream is precisely the relation collapse identified in \S\ref{sec:pilot}. Task-Interactive Modeling (TIM) alleviates this gap by adding a self-attention step over the task tokens.

\paragraph{Task self-attention.}
At each layer, the task stream applies self-attention using the same decoupled core $\mathrm{Attn}_{\mathrm{q}}$ in Equation \ref{eq:tcc}:
\begin{equation}
\hat{H}^{(\ell)}_{\mathrm{q}}
= \mathrm{Attn}_{\mathrm{q}}\!\big(
    H^{(\ell)}_{\mathrm{q}},\;
    H^{(\ell)}_{\mathrm{q}}
  \big).
\end{equation}
The task token representations are context-conditioned. Self-attention therefore operates on representations that are both \emph{task-identified} by the learnable embedding $e_k$ and \emph{session-grounded} by the accumulated cross-attention history. This completes an explicit interaction among task tokens that is learned from data rather than prescribed by a predefined task structure.

\paragraph{Causal visibility among tasks.}
The self-attention obeys the same visibility rule $c(p_k)\le c(p_q)$ introduced in \S\ref{sec:dsd}. \textit{via} token can attend to the \textit{where}, \textit{when} and \textit{how} tokens in the same session, conditioning the waypoint on the already-decided destination and mode, whereas the reverse is blocked. This ordering prevents label leakage while still allowing downstream facets to exploit upstream decisions.

\paragraph{Inter-Layer Connections.}
The attention calculation in \S\ref{sec:dsd} and \S\ref{sec:tim} are both computed from the layer input, and enter a single residual update.
\begin{equation}
\label{eq:layer-update}
H^{(\ell+1)}_{\mathrm{q}}
=H^{(\ell)}_{\mathrm{q}}
+\underbrace{\tilde H^{(\ell)}_{\mathrm{q}}}_{\text{context}}
+\underbrace{\hat H^{(\ell)}_{\mathrm{q}}}_{\text{cross-task}},
\qquad
H^{(\ell+1)}_{\mathrm{ctx}}
=H^{(\ell)}_{\mathrm{ctx}}+\tilde H^{(\ell)}_{\mathrm{ctx}},
\end{equation}

\subsection{HQ: Hierarchical Querying}
\label{sec:hq}

Section~\ref{sec:pilot} shows that the most informative depth differs across tasks and moves during training. A single final-layer read-out flattens these differences. Hierarchical Querying lets each task select the depth it needs, using only the task stream.

\paragraph{Per-layer collection.}
At every encoder layer $\ell$, the task-stream output is aggregated across residual paths, normalized, and stored.
\S\ref{sec:dsd} and \S\ref{sec:tim} yield a depth-indexed bank $\{h^{(1)}_{q}, \dots, h^{(L)}_{q}\}$ of task representations. We collect from the task stream rather than the context stream because the task-token representations have already been conditioned on each task's information needs through cross-attention. The depth aggregation in HQ therefore only needs to select which depth matters, not which information is task-relevant.

\paragraph{Depth attention.}
For task $k$ we stack the $L$ layer states collected above into $Z_{k}=\big[H^{(1)}_{\mathrm q,k};\dots;H^{(L)}_{\mathrm q,k}\big]
\in\mathbb R^{L\times d}$, and attend over its depth axis. The query is formed from the task identity, and each layer is mapped to a key and a value with a learnable depth embedding $d_\ell$ marking which layer it came from,
\begin{equation}
\label{eq:hq}
\left\{
\begin{aligned}
Q_k &= W_Q q_k, \qquad
K_\ell = W_K Z_k[\ell] + d_\ell, \qquad
V_\ell = W_V Z_k[\ell], \\
\alpha^{(k)}_\ell
&= \frac{\exp\!\big(\langle Q_k,K_\ell\rangle/\sqrt{d_a}\big)}
{\sum_{\ell'=1}^{L}
\exp\!\big(\langle Q_k,K_{\ell'}\rangle/\sqrt{d_a}\big)}, \\
z_k
&= \mathrm{LN}\!\left(
H^{(L)}_{\mathrm q,k}\,\Vert\,\bar z_k
\right),
\end{aligned}
\right.
\end{equation}  
where $q_k$ is the learnable embedding of task $k$ and $\bar z_k = \sum_{\ell=1}^{L}\alpha^{(k)}_\ell V_\ell$. The score in \eqref{eq:hq} depends on the task through $q_k$ and on the context of each layer through $\kappa_\ell$. The selected depth varies with the task, the session, and the user, rather than being fixed by a schedule.

\subsection{Training and Loss}
\label{sec:loss}

Each task is trained by its own InfoNCE loss~\citep{oord2018representation}. The output of Hierarchical Querying is passed through a task-specific head to obtain the final task representation. Note that \modelname is head-agnostic. We further prove out \modelname paradigm outperforms different generative backbone across different task heads in Table \ref{tab:main}. The loss maximizes the probability of the ground-truth item $i^{+}$ over $\mathcal{I}_k$. The final multi-task loss $L$ is the sum of the per-task losses, as shown in Eq.~\ref{eq:loss}. The overall training procedure is summarized in Algorithm~\ref{alg:model}.

\begin{algorithm}[ht]
\caption{\modelname forward and training step}
\label{alg:model}
\begin{algorithmic}[1]
\Require context embeddings $H^{(0)}_{\mathrm{ctx}}$, task embeddings
         $H^{(0)}_{\mathrm{q}}$ with $H^{(0)}_{\mathrm{q}}[t,k]=e_{\mathrm{cls}}+e_k$,
         layers $L$, decoupled cores
         $\mathrm{Attn}_{\mathrm{ctx}},\mathrm{Attn}_{\mathrm{q}}$,
         causal masks $M_{\mathrm{cc}},M_{\mathrm{qc}},M_{\mathrm{qq}}$
\Ensure  multi-task loss $\mathcal{L}$
\For{$\ell = 0$ \textbf{to} $L-1$}
  \Comment{the three calls below read the same layer input and run in parallel}
  \State $\tilde H_{\mathrm{ctx}} \leftarrow \mathrm{Attn}_{\mathrm{ctx}}\big(H^{(\ell)}_{\mathrm{ctx}},H^{(\ell)}_{\mathrm{ctx}};M_{\mathrm{cc}}\big)$
         \Comment{context self-attention}
  \State $\tilde H_{\mathrm{q}} \leftarrow \mathrm{Attn}_{\mathrm{q}}\big(H^{(\ell)}_{\mathrm{q}},H^{(\ell)}_{\mathrm{ctx}};M_{\mathrm{qc}}\big)$
         \Comment{task$\to$context cross-attention}
  \State $\hat H_{\mathrm{q}} \leftarrow \mathrm{Attn}_{\mathrm{q}}\big(H^{(\ell)}_{\mathrm{q}},H^{(\ell)}_{\mathrm{q}};M_{\mathrm{qq}}\big)$
         \Comment{task self-attention}
  \State $H^{(\ell+1)}_{\mathrm{ctx}} \leftarrow H^{(\ell)}_{\mathrm{ctx}}+\tilde H_{\mathrm{ctx}}$
         \Comment{residual update, task-free}
  \State $H^{(\ell+1)}_{\mathrm{q}} \leftarrow H^{(\ell)}_{\mathrm{q}}+\tilde H_{\mathrm{q}}+\hat H_{\mathrm{q}}$
  \State $h^{(\ell+1)}_{\mathrm{q}} \leftarrow \mathrm{LN}^{(\ell+1)}\big(H^{(\ell+1)}_{\mathrm{q}}\big)$
         \Comment{keep for hierarchical querying}
\EndFor
\State $\mathcal{L} \leftarrow 0$
\For{each task $k \in \mathcal{K}$}
  \State $\bar z_k \leftarrow
  \mathrm{DepthAttn}\big(e_k,
  \{h^{(\ell)}_{\mathrm q}[\cdot,k]\}_{\ell=1}^{L}\big)$
         \Comment{weighted aggregation over all query-stream layers}
  \State $z_k \leftarrow
  \mathrm{Concat}\big(h^{(L)}_{\mathrm q}[\cdot,k],\bar z_k\big)$
         \Comment{concatenate the final-layer and hierarchical representations}
  \State $o_k \leftarrow \mathrm{Head}_k(z_k)$
  \State $\mathcal{L} \leftarrow \mathcal{L} + \mathrm{InfoNCE}\big(o_k,i^{+};\mathcal{I}_k\big)$
         \Comment{Eq.~\ref{eq:loss}}
\EndFor
\State update parameters by $\nabla\mathcal{L}$; \Return $\mathcal{L}$
\end{algorithmic}
\end{algorithm}

We build the candidate set $\mathcal{I}_k$ per task. For a task with a small label space, such as \emph{how} and \emph{when}, $\mathcal{I}_k$ is the full class set and the loss is a standard softmax cross entropy. For a task with a very large label space, such as \emph{where} and \emph{via} over POIs, we replace the full softmax by a sampled softmax, so $\mathcal{I}_k$ is the ground-truth item together with a set of sampled negatives. \modelname, including the two streams, the task interaction, and the hierarchical query, is trained end to end by minimizing $L$.

\begin{equation}
\label{eq:loss}
L=\sum_{k\in\mathcal{K}}-\frac{1}{|\mathcal{A}|}\sum_{u\in\mathcal{U}}\sum_{a\in\mathcal{A}_u}
\log\frac{\exp(\hat{y}^{\,i^{+}}_k)}{\sum_{i\in\mathcal{I}_k}\exp(\hat{y}^{\,i}_k)} .
\end{equation}

%% file: tex/exp.tex
\section{Experiment}

\subsection{Diagnosing Hierarchical Collapse}
\label{sec:pilot-hier}

\begin{figure}[ht]
\centering
\includegraphics[width=0.8\linewidth]{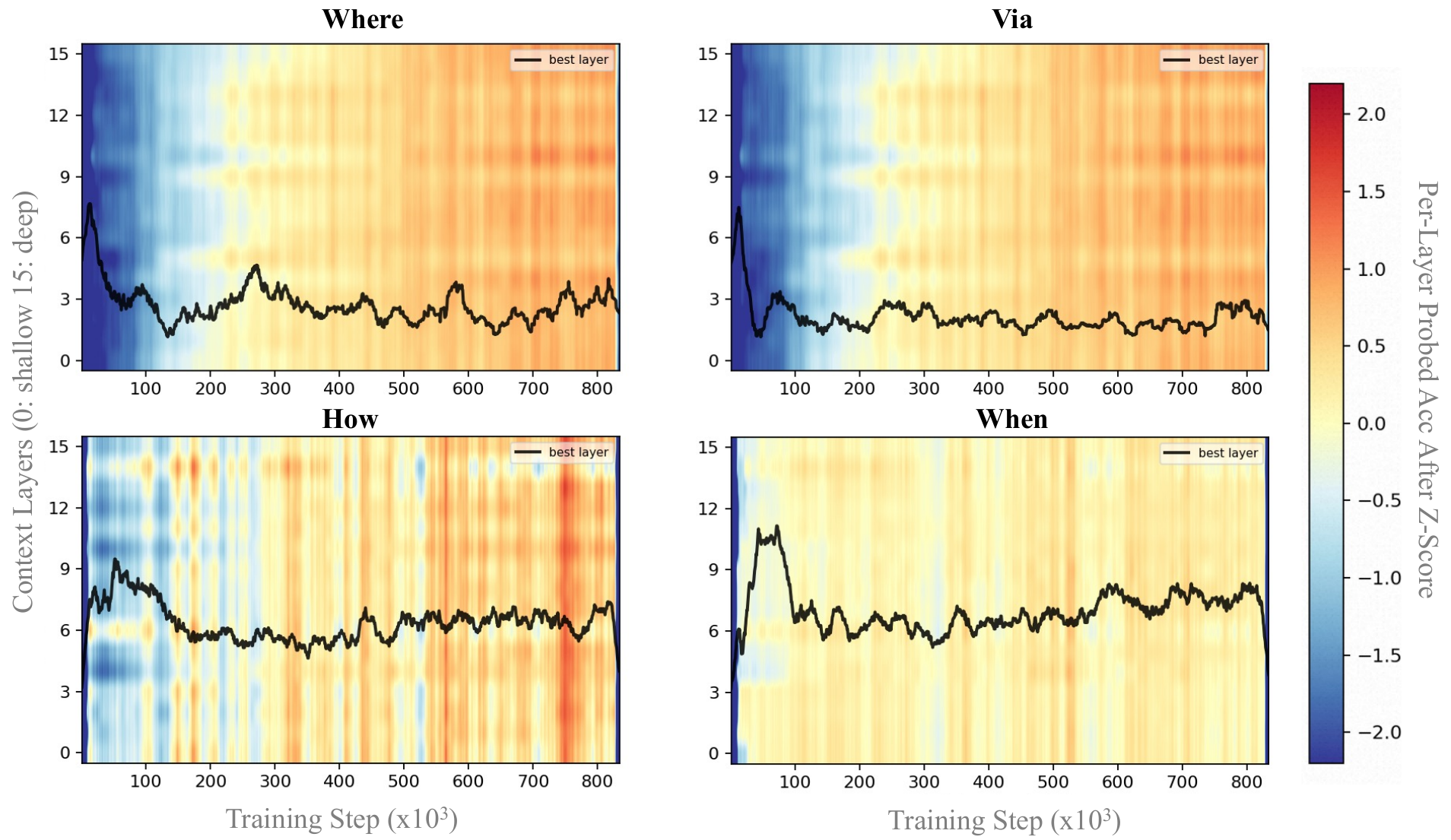}
\caption{\textbf{Per-layer probed accuracy over training}. Color is the per-layer probed accuracy $z$-scored along training (warm${=}$high, cool${=}$low). The black curve is the smoothed best layer $\arg\max_\ell$.For different tasks, the reliance on features of different scales varies at different stages of training.}
\label{fig:probe}
\end{figure}

To align with our online system, we train an $L{=}16$ layers baseline model, implemented as IntTravel~\citep{yan2026inttravel}, on production data \footnote{The production data follows the same data format as the public benchmark~\citep{yan2026inttravel}, and differs only in scale.} and attach per-layer linear probing network. For each layer $\ell$, a probe predicts $y_k$ from $z_\ell$. Figure~\ref{fig:probe} reports the per-layer probed accuracy over training. The most-decodable layer differs systematically across tasks. Each task's best layer (black curve) drifts throughout training, so no fixed depth is optimal even for one task across time. Fine-grained retrieval (\emph{Where}, \emph{Via}) peaks at shallow layers ($\ell\!\approx\!1$--$3$), whereas coarse intent (\emph{How}, \emph{When}) relies on \emph{deeper}, more abstract features ($\ell\!\approx\!5$--$7$). Intuitively, depth is governed by the range of context each task consumes. Shallow layers preserve high-resolution, immediate signals from the current session, while deeper layers progressively aggregate behaviors into long-range patterns, and tasks benefit from the two ends differently. Travel-mode prediction (\emph{How}) profits from deep features that accumulate a user's commuting regularities across many sessions. Destination retrieval (\emph{Where}) hinges on the session-level scenario. The immediate spatio-temporal context exposes in shallow, high-resolution features.

\subsection{Experiment Settings}
\subsubsection{Datasets.}

We adopt IntTravel~\citep{yan2026inttravel}, a large-scale dataset built from logs of a real-world industrial map and navigation platform. It contains about $4.1$ billion interactions from $163$ million users over $7.3$ million POIs, which is two to three orders of magnitude larger than prior benchmarks~\citep{cheng2011exploring,cho2011friendship,yang2014modeling,yang2016participatory,monti2018semantic,yang2019revisiting}, and it covers four coupled tasks: \emph{when} to depart, \emph{where} to go, \emph{how} to travel, and \emph{via} which on-the-way POI to visit.

   \begin{table}[ht]
    \centering
    \caption{Offline performance comparison of four encoders paired with four multi-task heads
    on the IntTravel dataset. $\uparrow$ indicates higher is better and $\downarrow$ lower is better. Within each multi-task head, the best encoder is highlighted with
    a \colorbox{dreamblue!15}{blue background} and the second best is in \textcolor{dreamblue}{blue}. $\dagger$ indicates a statistically
significant improvement over the strongest baseline under the same head ($t$-test, $p$-value $<$ 0.05).}
    \label{tab:main}
    \vspace{6pt}
    \setlength{\tabcolsep}{3pt}
    \resizebox{\textwidth}{!}{%
    \begin{tabular}{ll cc cc ccc ccc}
    \toprule
    \multirow{2}{*}{Encoder} & \multirow{2}{*}{Head}
    & \multicolumn{2}{c}{\emph{When}}
    & \multicolumn{2}{c}{\emph{How}}
    & \multicolumn{3}{c}{\emph{Where}}
    & \multicolumn{3}{c}{\emph{Via}} \\
    \cmidrule(lr){3-4} \cmidrule(lr){5-6}
    \cmidrule(lr){7-9} \cmidrule(lr){10-12}
    & & Acc$\uparrow$ & MAE$\downarrow$
    & Acc$\uparrow$ & BCR$\downarrow$
    & HR@1$\uparrow$ & HR@5$\uparrow$ & CIR$\downarrow$
    & HR@1$\uparrow$ & HR@5$\uparrow$ & CIR$\downarrow$ \\
    \midrule

    \multirow{4}{*}{IntTravel}
    & PLE
    & 0.8338 & 7.9784 & \second{0.6827} & 0.0678
    & 0.6796 & 0.8676 & 0.2512
    & 0.6806 & 0.8678 & 0.2565 \\
    & STAR
    & \second{0.8335} & \second{7.9942}
    & \second{0.6822} & \second{0.0680}
    & \second{0.7033} & \second{0.8799} & \second{0.2297}
    & \second{0.7122} & \second{0.8824} & \second{0.2294} \\
    & DSFNet
    & \second{0.8336} & \second{7.9885}
    & \second{0.6819} & \second{0.0682}
    & \second{0.7051} & \second{0.8808} & \second{0.2271}
    & \second{0.7141} & \second{0.8833} & \second{0.2275} \\
    & HoME
    & \second{0.8335} & \second{7.9933}
    & \second{0.6828} & \second{0.0677}
    & \second{0.6823} & \second{0.8688} & 0.2489
    & 0.6858 & \second{0.8707} & 0.2517 \\

    \midrule
    \multirow{4}{*}{OneTrans}
    & PLE
    & \second{0.8342} & \second{7.9593}
    & 0.6821 & \second{0.0678}
    & \second{0.6822} & \second{0.8686} & \second{0.2486}
    & \second{0.6845} & \second{0.8694} & \second{0.2531} \\
    & STAR
    & 0.8330 & 8.0175 & 0.6682 & 0.0724
    & 0.6961 & 0.8771 & 0.2316
    & 0.7099 & 0.8815 & 0.2321 \\
    & DSFNet
    & 0.8328 & 8.0236 & 0.6739 & 0.0712
    & 0.6997 & 0.8776 & 0.2313
    & 0.7105 & 0.8815 & 0.2318 \\
    & HoME
    & 0.8330 & 8.0163 & 0.6734 & 0.0704
    & 0.6790 & 0.8681 & \second{0.2479}
    & \second{0.6870} & 0.8706 & \second{0.2515} \\

    \midrule
    \multirow{4}{*}{HGenPush}
    & PLE
    & 0.8310 & 8.1102 & 0.5714 & 0.1215
    & 0.6623 & 0.8555 & 0.2656
    & 0.6795 & 0.8666 & 0.2568 \\
    & STAR
    & 0.8310 & 8.1115 & 0.5739 & 0.1206
    & 0.6874 & 0.8680 & 0.2434
    & 0.7078 & 0.8797 & 0.2339 \\
    & DSFNet
    & 0.8311 & 8.1090 & 0.5728 & 0.1208
    & 0.6855 & 0.8664 & 0.2443
    & 0.7069 & 0.8788 & 0.2348 \\
    & HoME
    & 0.8310 & 8.1113 & 0.5746 & 0.1203
    & 0.6627 & 0.8563 & 0.2649
    & 0.6803 & 0.8673 & 0.2568 \\

    \midrule
    \multirow{4}{*}{IntHQ (Ours)}
    & PLE
    & \best{0.8349} & \best{7.9271}
    & \best{0.6929}$^{\dagger}$ & \best{0.0642}$^{\dagger}$
    & \best{0.6879}$^{\dagger}$ & \best{0.8726}$^{\dagger}$ & \best{0.2429}$^{\dagger}$
    & \best{0.6864}$^{\dagger}$ & \best{0.8716}$^{\dagger}$ & \best{0.2524}$^{\dagger}$ \\
    & STAR
    & \best{0.8349} & \best{7.9246}
    & \best{0.6946}$^{\dagger}$ & \best{0.0637}$^{\dagger}$
    & \best{0.7080}$^{\dagger}$ & \best{0.8824}$^{\dagger}$ & \best{0.2255}$^{\dagger}$
    & \best{0.7148}$^{\dagger}$ & \best{0.8842}$^{\dagger}$ & \best{0.2278}$^{\dagger}$ \\
    & DSFNet
    & \best{0.8352} & \best{7.9084}
    & \best{0.6947}$^{\dagger}$ & \best{0.0636}$^{\dagger}$
    & \best{0.7090}$^{\dagger}$ & \best{0.8827}$^{\dagger}$ & \best{0.2242}$^{\dagger}$
    & \best{0.7157}$^{\dagger}$ & \best{0.8842}$^{\dagger}$ & \best{0.2265}$^{\dagger}$ \\
    & HoME
    & \best{0.8350} & \best{7.9203}
    & \best{0.6915}$^{\dagger}$ & \best{0.0648}$^{\dagger}$
    & \best{0.6882}$^{\dagger}$ & \best{0.8725}$^{\dagger}$ & \best{0.2424}$^{\dagger}$
    & \best{0.6894}$^{\dagger}$ & \best{0.8724}$^{\dagger}$ & \best{0.2495}$^{\dagger}$ \\

    \bottomrule
    \end{tabular}}
    \end{table}

\subsubsection{Baselines.}
We organize baselines along two orthogonal axes: the \emph{sequence encoder}, which determines whether and how task signals enter the encoding computation, and the \emph{multi-task head}, which determines how the encoded representation is shared and specialized across tasks.

\paragraph{Sequence encoders.}
We compare four encoders that form a progression of task-signal injection. \textbf{IntTravel}~\citep{yan2026inttravel} is a HSTU-style backbone tailored to travel recommendation. It encodes the session-structured sequence without any task token, and tasks branch off only after encoding. \textbf{OneTrans}~\citep{zhang2026onetrans} unifies sequential and non-sequential features into one token sequence with mixed parameterization. \textbf{HGenPush}~\citep{liang2026hgenpush} takes a first step toward task-aware encoding. Following decoder-only organization of user behaviors and per-branch $[cls]$ tokens, task tokens are inserted into the main sequence, but they share one stream and one set of parameters with the context tokens. 

\paragraph{Multi-task heads.}
Each encoder is paired with four representative heads that cover the mainstream design space of multi-task specialization. \textbf{PLE}~\citep{tang2020progressive} routes shared and task-specific experts through progressive layered extraction to mitigate negative transfer. \textbf{STAR}~\citep{sheng2021one} multiplies a shared centered network with task-specific weights in a star topology. \textbf{DSFNet}~\citep{yu2024dsfnet} is our production head with per-task private experts on top of a shared expert pool. \textbf{HoME}~\citep{wang2025home} organizes hierarchical mixture-of-experts with feature gating and self-gated residuals for grouped tasks. 

\subsubsection{Evaluation Metrics.}
We follow the evaluation protocol of IntTravel~\citep{yan2026inttravel}, pairing each task with positive indicators and a negative indicator to capture both prediction quality and user-experience cost. For the \emph{when} task, we report Accuracy (Acc$\uparrow$), the exact-match rate of the predicted departure time, and Mean Absolute Error (MAE$\downarrow$) as the average deviation from the ground truth.
For the \emph{how} task, we report Acc$\uparrow$ over travel modes and the Bad Case Rate (BCR$\downarrow$), defined as the top-3 miss rate, since navigation apps display only three modes on the primary screen. For the \emph{where} and \emph{via} tasks, we report HitRate@\{1,5\}$\uparrow$ over the candidate POIs and the Category Inconsistency Rate (CIR$\downarrow$), the rate at which the top-1 prediction neither matches the ground-truth POI nor shares its category.

\subsubsection{Implementation Details.}
All models share identical training configurations. The embedding dimension is $96$, the maximum sequence length to 120, and the batch size is $64$ per worker on $8$ PPU GPUs. All encoders are stacked with $16$ layers. Models are trained with AdamW (learning rate $8\times10^{-4}$, weight decay $10^{-6}$) for one epoch. The POI-retrieval tasks (\emph{where}, \emph{via}) are trained with sampled softmax over $64$ negatives per positive: $14$ drawn uniformly from the full POI vocabulary, and $50$ distance-based hard negatives. The context and task streams use fully decoupled attention parameters, and the hierarchical query attention uses dimension $64$ with a single head. Baseline encoders (OneTrans, HGenPush) use $4$-head attention with an FFN. Baseline heads adopt the hyperparameters reported in their original papers, adapted to the shared embedding dimension. All encoder--head combinations are trained with the same data order and random seeds for comparability.

\subsection{Overall Performance Comparison}

We evaluate \modelname and all baselines on the industrial-scale IntTravel dataset, with the full encoder--head grid reported in Table~\ref{tab:main}. Under every head, \modelname outperforms all encoder baselines on all ten metrics, improving positive and negative indicators simultaneously. The gains are a uniform expansion of the quality frontier rather than a trade-off among tasks. As for the encoder axis, IntTravel strengthens the backbone but carries no task token and no task interaction. Tasks branch off only after a task-agnostic encoding, leaving both collapses intact. OneTrans unifies all features into a single stream, so task signals never enter the encoder, and its single-depth read-out serves every task from the same layer. The two ends of the depth spectrum are then in direct competition, and the loss is dominated by the large-vocabulary tasks, so the coarse-intent tasks \emph{How} and \emph{When}, which rely on abstract long-range features, degrade the most. HGenPush does insert task tokens, but share one stream, as source collapse predicts, the two roles dilute each other, and it performs worst precisely on the same task (\emph{How}). \modelname completes the progression. Through early task injection in a decoupled stream, explicit task interaction, and hierarchical read-out, \modelname achieves the largest margins exactly where the baselines fail. Meanwhile, within \modelname the performance spread across the four heads narrows markedly, encoder-side task conditioning absorbs most of the specialization work, confirming that the advantage stems from the structure itself rather than any particular multi-task head. Notably, these improvements come at an acceptable computational cost. We detail the comparison in Appendix~\ref{app:cc}, taking PLE as a representative task head for illustration.

\subsection{Ablation Study}

  \begin{table}[t]
  \centering
  \caption{Ablation studies of \modelname on four tasks. We report Acc for \emph{when}/\emph{how} and
  HR@1 for \emph{where}/\emph{via}.}
  \label{tab:ablation}
  \setlength{\tabcolsep}{3pt}
  \begin{tabular}{l cccc}
  \toprule
  Task (Metric) & \modelname (full) & w/o DSD & w/o TIM & w/o HQ \\
  \midrule
  \emph{When} (Acc$\uparrow$)
  & \best{0.8352} & 0.8324 & 0.8345 & 0.8349 \\
  \emph{Where} (HR@1$\uparrow$)
  & \best{0.7090} & 0.6796 & 0.7086 & 0.7076 \\
  \emph{How} (Acc$\uparrow$)
  & \best{0.6947} & 0.6621 & 0.6923 & 0.6933 \\
  \emph{Via} (HR@1$\uparrow$)
  & \best{0.7157} & 0.6954 & 0.7153 & 0.7147 \\
  \bottomrule
  \end{tabular}
  \end{table}

Table~\ref{tab:ablation} ablates the three components of \modelname, each of which proves necessary, and the pattern of degradation aligns with the collapse each component targets. Replacing \textbf{DSD} as the shared encoder causes significant and uniform drops across all four tasks, as expected for source collapse, whose dilution effect is global rather than task-specific. Removing \textbf{TIM} leaves the driver tasks (\emph{When}, \emph{Where}) nearly untouched but visibly hurts \emph{How} and \emph{Via}, the facets whose decisions are conditioned on other tasks. This asymmetry is direct evidence that the gains of explicit task interaction come from modeling inter-task coupling. Removing \textbf{HQ} produces degradation on four tasks, confirming that a single read-out depth cannot serve heterogeneous granularity.

\subsection{Scaling Law}

A central promise of generative recommendation is that performance scales with compute. We verify that \modelname inherits this property by sweeping the encoder depth $L\in\{2,4,8,16,32\}$ and width $D\in\{24,48,96,192,384\}$ with all other hyper-parameters fixed, training each variant on the same data stream. Figure~\ref{fig:scaling} shows related metrics across four tasks. Deeper and wider encoders present better performance as expected. The dual-stream structure scales without modification. Adding model capacity enlarges both the context stream's capacity for long-range patterns and the bank of representations that Hierarchical Querying can draw from.

\begin{figure}[ht]
\centering
\includegraphics[width=\linewidth]{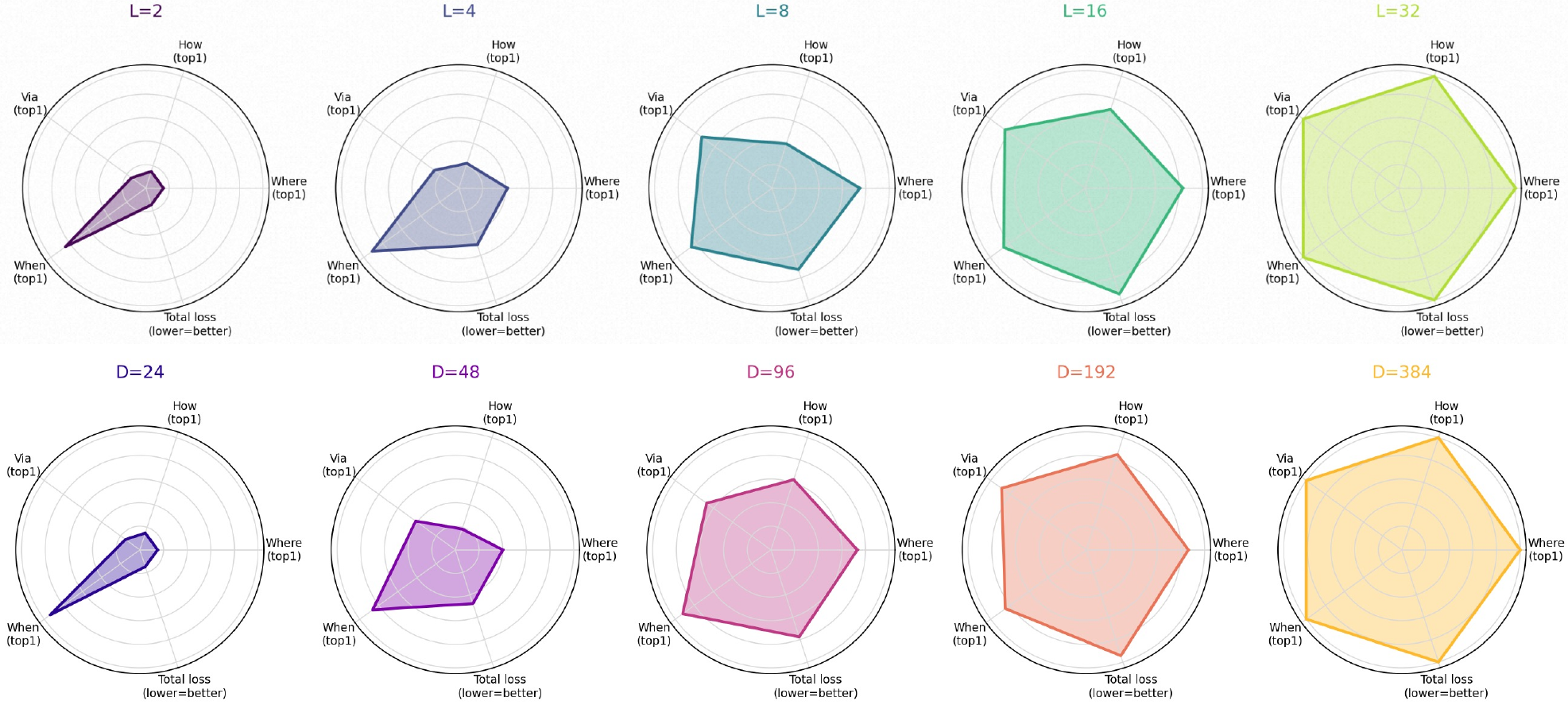}
\caption{Scaling with encoder depth and embedding dimension. Total loss, accuracy and hit rate of \modelname with $L\in\{2,4,8,16,32\}$ layers and $D\in\{24,48,96,192,384\}$ dimensions. Deeper and wider encoders converge to better performance, exhibiting clear depth-wise and width-wise scaling trends.}
\label{fig:scaling}
\end{figure}

\subsection{Online Deployment}

\subsubsection{Deployment Pipeline.}
Serving \modelname in production must reconcile two competing demands. The model has to track travel behavior that shifts daily while every request must be answered within a tight latency budget. Figure~\ref{fig:online} shows how we resolve them by moving all heavy, user-level computation off the request path. Sequence assembly is precomputed rather than performed online. An offline job materializes each user's session structured history from raw behavior logs into a key-value store. At request time the model service retrieves this prefix by user ID. The assembled sequence is scored on GPU computation nodes. For \emph{where} and \emph{via} the resulting task representations are matched against the POI index by approximate nearest neighbor search. Response generated by \modelname service and user feedback are written back to the log store, from which the next training round is assembled and the refreshed model is republished to the serving engines. Overall, \modelname handles a sequence length of 300, covering 100 most recent behavioral interactions, under a peak concurrency of 30k QPS (query per second). Deployed on a cluster of 150 NVIDIA PPU GPUs (96 GB memory), the average response time (RT) is 40 ms and the 99\% RT is 100 ms.

\begin{figure}[ht]
\centering
\includegraphics[width=0.65\linewidth]{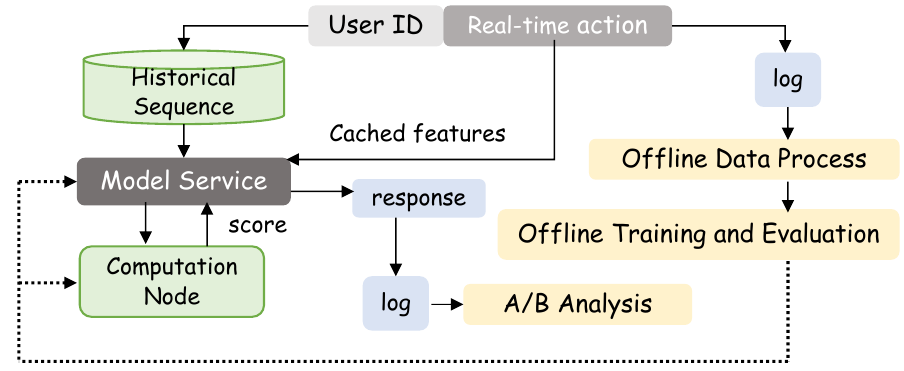}
\caption{Online deployment architecture of \modelname on Amap. Solid arrows denote the real-time serving flow. Dashed arrows denote the offline training loop that redeploys updated models to the serving engines.}
\label{fig:online}
\end{figure}

\subsubsection{Online A/B Results.}
We evaluated \modelname against the online baseline in the travel recommendation scenario of Amap, with each bucket receiving 20\% of live traffic over 7 days. \modelname improved UVCTR by 1.60\% relative. The gain held consistently across daily slices. Handling 30k QPS on a cluster of PPU GPUs, the average response time (RT) of \modelname is 40 ms. \modelname has since been fully rolled out and now serves all traffic in this scenario.

%% file: tex/related_work.tex
\section{Related Work}

\subsection{Classical MTL}
Early shared-bottom networks and Cross-Stitch~\citep{misra2016cross} share a single trunk and branch into per-task towers. Expert-based designs
refine this sharing. MMoE~\citep{ma2018modeling} routes soft-selected experts per task, and PLE~\citep{tang2020progressive} explicitly separates shared and task-specific experts through progressive layered extraction to curb negative transfer. A parallel line targets multi-domain/multi-scenario settings. STAR~\citep{sheng2021one} multiplies a shared centered network with domain-specific weights in a star topology, while M2M~\citep{zhang2022leaving} and APG~\citep{yan2022apg} generate customized network parameters from scenario or instance context. Recent mixture-of-experts variants such as STEM~\citep{su2024stem} and HoME~\citep{wang2025home} push specialization further via task-specific embeddings, stop-gradient gating, and hierarchical expert grouping. Orthogonally, an optimization-centric line attributes multi-task interference to conflicting gradients on shared parameters and manipulates the update direction, e.g., PCGrad~\citep{yu2020gradient}, GradNorm~\citep{chen2018gradnorm}, CAGrad~\citep{liu2021conflict} and Nash-MTL~\citep{navon2022multi}.

\subsection{MTL in GR}
Generative recommendation (GR) reformulates recommendation as sequential transduction over user behaviors, with backbones such as HSTU~\citep{zhai2024actions}, TIGER~\citep{rajput2023recommender}, and OneRec~\citep{deng2025onerec} scaling attention-based next-item generation. Multi-task modeling in GR, however, remains largely inherited from the classical recipe. Tasks are still read out of one shared sequence representation. IntTravel~\citep{yan2026inttravel} strengthens an HSTU backbone for travel recommendation but encodes the behavior sequence without task tokens, branching only after encoding. OneTrans~\citep{zhang2026onetrans} unifies sequential and non-sequential features into a single token stream with mixed shared/token-specific parameterization, yet it likewise carries no task token and follows the unified single-stream paradigm. HGenPush~\citep{liang2026hgenpush} and Onerank~\citep{tang2026onerank} takes a first step toward task-aware encoding by inserting task tokens into the sequence, but these tokens share one stream and one parameter set with the context tokens, task-specific signals are still diluted in the shared space.

%% file: tex/conclusion.tex
\section{Conclusion}
We revisit multi-task generative recommendation through the lens of a threefold collapse inherent in the prevailing recipe of a single task-agnostic representation with conditionally independent heads. We propose \modelname, which counters each collapse with a dedicated component. Dual-Stream Decoupling injects task signals early and separates the task stream from the context stream in parameter space. Task-Interactive Modeling learns input-dependent task interactions rather than a prescribed structure. Hierarchical Querying lets each task retrieve information across representations at different scales. In online A/B tests on Amap serving hundreds of millions of users, \modelname delivers a 1.60\% UVCTR relative lift. Experiments across a grid of encoders and multi-task heads confirm that the gains stem from moving task signals into the encoding stage and persist across heterogeneous head architectures. This work offers a practical way to scale generative recommender systems in complex multi-task settings.

%% file: tex/appendix.tex
\section{Proofs and Extended Statements for the Pilot Study}
\label{app:proofs}

\subsection{Source collapse motivates Dual-Stream Decoupling.}
Two statements settle the task signal must enter and on which parameters it must live. The first rules out the standard remedy.

\begin{proposition}[Late conditioning is powerless]
\label{prop:late}
If every task-specific computation is a function of the shared code, $h_k=g_k(z)$, then $\sum_k I(\mathcal A;y_k\mid h_k)\ \ge\ \Delta_{\mathrm{src}}(z)$, with equality if and only if each $g_k$ is sufficient for $y_k$ relative to $z$.
\end{proposition}

Nothing placed above $z$ can reduce the source term, because each module only coarsens a statistic from the discarded information. Task conditioning must therefore happen inside the encoder. Proposition \ref{prop:source} shows that placing it inside a stream whose weights are shared with the context computation is still not enough.

\begin{proposition}[Price of a shared block]
\label{prop:source}
Let a weight block act in two roles, and write the risk on two copies of its parameters as $R(\Theta_c,\Theta_q)=L_c(\Theta_c)+L_q(\Theta_q)$, the additive form of the gradient split $\nabla_\Theta R=g_c+g_q$. Let $L_c,L_q$ have minimizers $\Theta_c^\ast,\Theta_q^\ast$ with positive-definite Hessians $H_c,H_q$, and let $\Theta^t$ minimize the tied risk $L_c(\Theta)+L_q(\Theta)$. Then at $\Theta^t$ the two roles pull with equal and opposite force,
$\nabla L_c(\Theta^t)=-\nabla L_q(\Theta^t)$, so neither role sits at its own optimum, and the tied risk exceeds the untied risk $L_c(\Theta_c^\ast)+L_q(\Theta_q^\ast)$ by
\begin{equation}
\label{eq:price}
\tfrac12\,\Delta^\top H_c\,(H_c+H_q)^{-1}H_q\,\Delta\ \ge\ 0,
\qquad \Delta:=\Theta_c^\ast-\Theta_q^\ast,
\end{equation}
which vanishes if and only if the two role optima coincide.
\end{proposition}

The cost \eqref{eq:price} depends only on how far apart the two roles want the block to be, not on how many parameters it has. When they disagree, sharing forces a curvature-weighted compromise that serves neither. Dual-Stream Decoupling answers both propositions. It conditions the computation on the task from the first layer, so the task identity is present before the compression is complete, and it keeps the task stream on parameter-decoupled weights, so each role reaches its own optimum. This turns the state and prediction separation of~\citep{monea2026state}, established for single-task language modeling, into multi-task case with $K$ task roles against one context role.

\subsection{Relation collapse motivates Task-Interactive Modeling.}

\begin{proposition}[Factorization gap and the escape routes]
\label{prop:relation}
On a fixed encoder, replacing the best factorized read-out by the best joint read-out reduces the risk by exactly $\mathrm{TC}(\mathbf y\mid z)$, and
\begin{equation}
\label{eq:tc-relation}
\mathrm{TC}(\mathbf y\mid z)-\mathrm{TC}(\mathbf y\mid\mathcal A)
=\sum_k I(\mathcal A;y_k\mid z)-I(\mathcal A;\mathbf y\mid z),
\end{equation}
so the realized cost and the encoder-free floor coincide only at joint sufficiency. Attaining $R^\star$ requires a read-out that is not factorized given $\mathcal A$, and there are exactly three ways to obtain one. An explicit joint head over $\prod_k\mathcal I_k$, a latent-variable head scored by its marginal likelihood, or a chain-rule factorization in which each decision conditions on realized predecessor outcomes.
\end{proposition}

\begin{proposition}[Residual under realized-outcome conditioning]
\label{prop:escape}
Let $\prec$ be a linear extension of the funnel order and let each head condition on the realized labels of $\mathrm{pa}(k)\subseteq\{y_j:j\prec k\}$. Then the optimal excess risk over $R^\star$ equals
\begin{equation}
\label{eq:residual}
\sum_k I\big(y_k;\,y_{\prec k}\setminus\mathrm{pa}(k)\ \big|\ y_{\mathrm{pa}(k)},\mathcal A\big),
\end{equation}
which is invariant to the choice of extension, vanishes when every parent set is the full predecessor set, and equals $\Delta_{\mathrm{rel}}$ when all parent sets are empty.
\end{proposition}

Task-Interactive Modeling places each decision at the position where its outcome is realized inside a session, and masks the reverse direction. Funnel models such as ESMM~\citep{ma2018entire} and AITM~\citep{xi2021modeling} also condition on predecessors, yet they fix the order and the coupling strength in advance, whereas cross-task attention learns that strength for each instance.

\subsection{Gibbs optimum and the value of \texorpdfstring{$R_{\mathrm{TAE}}$}{RTAE}}

\begin{lemma}[Gibbs inequality~\citep{cover1999elements}]
\label{lem:gibbs}
Let $y$ be a target with finitely possible values and let $u$ collect everything the predictor is allowed to observe. Write $p(\cdot\mid u)$ for the true conditional distribution of $y$ given $u$, and let $q(\cdot\mid u)$ be any predicted distribution that gives positive probability to every value, as a softmax head does. Then the expected log loss of $q$ splits into two parts,
\begin{equation}
\label{eq:gibbs}
\underbrace{\mathbb E\big[-\log q(y\mid u)\big]}_{\text{loss of the prediction}}
=\underbrace{H(y\mid u)}_{\text{uncertainty of the truth}}
+\underbrace{\mathbb E_u\,D_{\mathrm{KL}}\big(p(\cdot\mid u)\,\|\,q(\cdot\mid u)\big)}_{\text{distance from the truth}}
\;\ge\;H(y\mid u).
\end{equation}
The second part is zero exactly when $q(\cdot\mid u)=p(\cdot\mid u)$. The log loss is therefore minimized by predicting the true conditional distribution and by nothing else, which is the statement that it is a \emph{strictly proper} scoring rule~\citep{gneiting2007strictly}.                                       \end{lemma}

\begin{proof}
We give the short derivation for completeness. For one outcome, add and subtract the log probability of the truth,
\[
-\log q(y\mid u)=-\log p(y\mid u)+\log\frac{p(y\mid u)}{q(y\mid u)} .
\]
Now average both sides over $y$ and $u$. The first term averages to $H(y\mid u)$, which is the definition of conditional entropy, and the second averages to $\mathbb E_u D_{\mathrm{KL}}(p\,\|\,q)$, which is the definition of the divergence. This gives the identity in \eqref{eq:gibbs}.
\end{proof}

\begin{lemma}[Optimal risk of a factorized read-out]
\label{lem:rtae}
Suppose each head is free to output any distribution over its own candidate set and no parameter is shared between heads, the $K$ heads can be chosen independently of one another. Then
\begin{equation}
\label{eq:rtae}
R_{\mathrm{TAE}}(z):=\inf_{q_1,\dots,q_K}\ \mathbb E\Big[-\sum_k\log q_k(y_k\mid z)\Big]
=\sum_k H(y_k\mid z),
\end{equation}
reached by letting every head output the true conditional distribution $q_k=p(y_k\mid z)$. Since $-\sum_k\log q_k(y_k\mid z)=-\log\prod_k q_k(y_k\mid z)$, the left side is the joint log loss of the product law $\prod_k q_k(\cdot\mid z)$, which is what makes $R_{\mathrm{TAE}}(z)$ comparable with $R^\star$.
\end{lemma}

\subsection{Proof of Theorem~\ref{thm:unified}}

\begin{proof}
Both $R^\star$ and $R_{\mathrm{TAE}}(z)$ are values of the same joint log loss: Lemma~\ref{lem:gibbs} with target $\mathbf y$ and $u=\mathcal A$ gives $R^\star=H(\mathbf y\mid\mathcal A)$, and Lemma~\ref{lem:rtae} gives $R_{\mathrm{TAE}}(z)=\sum_k H(y_k\mid z)$.

Because $z$ is a deterministic function of $\mathcal A$, conditioning on $(z,\mathcal A)$ is the same as conditioning on $\mathcal A$, so $H(y_k\mid z)=H(y_k\mid\mathcal A)+I(\mathcal A;y_k\mid z)$; all terms are finite. Summing over $k$ and using the definition of total correlation, $\sum_k H(y_k\mid\mathcal A)=H(\mathbf y\mid\mathcal A)+\mathrm{TC}(\mathbf y\mid\mathcal A)$,
\[
R_{\mathrm{TAE}}(z)
=\underbrace{\sum_k I(\mathcal A;y_k\mid z)}_{\Delta_{\mathrm{src}}(z)}
+H(\mathbf y\mid\mathcal A)
+\underbrace{\mathrm{TC}(\mathbf y\mid\mathcal A)}_{\Delta_{\mathrm{rel}}},
\]
which is \eqref{eq:unified} after subtracting $R^\star=H(\mathbf y\mid\mathcal A)$. Both terms are non-negative, $\Delta_{\mathrm{src}}(z)$ as a sum of conditional mutual informations and $\Delta_{\mathrm{rel}}=\mathbb E_{\mathcal A}D_{\mathrm{KL}}\big(P(\mathbf y\mid\mathcal A)\,\|\,\prod_k P(y_k\mid\mathcal A)\big)$ as an expected divergence. Dropping $\Delta_{\mathrm{src}}(z)$ gives the floor
$R_{\mathrm{TAE}}(z)\ge R^\star+\Delta_{\mathrm{rel}}$, with equality exactly when $I(\mathcal A;y_k\mid z)=0$ for all $k$, that is when $z$ is marginally sufficient for every task.
\end{proof}

\subsection{Source term}

\subsubsection{Late conditioning cannot reduce the source term}

\begin{proof}[Proof of Proposition~\ref{prop:late}]
Let $h_k=g_k(z)$ with $g_k$ measurable. We show the inequality term by term.

First, $\sigma(h_k)\subseteq\sigma(z)$, because $h_k$ is a measurable function of
$z$. Conditioning on a coarser $\sigma$-field cannot decrease conditional entropy,
hence
\begin{equation}
\label{eq:coarse}
H(y_k\mid h_k)\ \ge\ H(y_k\mid z).
\end{equation}
Second, $h_k$ is also a measurable function of $\mathcal A$, since $z$ is, so Step 2
of the proof of Theorem~\ref{thm:unified} applies with $h_k$ in place of $z$ and
gives $I(\mathcal A;y_k\mid h_k)=H(y_k\mid h_k)-H(y_k\mid\mathcal A)$. Combining this
with \eqref{eq:coarse} and with Step 3 of that proof,
\[
I(\mathcal A;y_k\mid h_k)
=H(y_k\mid h_k)-H(y_k\mid\mathcal A)
\ \ge\ H(y_k\mid z)-H(y_k\mid\mathcal A)
=I(\mathcal A;y_k\mid z).
\]
Summing over $k$ gives the claim. Equality holds for a given $k$ if and only if
\eqref{eq:coarse} is tight, that is if and only if $h_k$ retains all the information
about $y_k$ that $z$ carries.
\end{proof}

\subsubsection{The price of a shared block}

A weight block $W$ that a Transformer applies at every position acts in two roles. At context positions it prepares keys and values read by later positions, the context role; at task positions it produces the read-out for a decision, the task role. By the chain rule its gradient splits by source,
\begin{equation}
\label{eq:grad-split}
\nabla_W R=\underbrace{\textstyle\sum_{p\in\mathcal C}\nabla_{W_p}R}_{g_c}
+\underbrace{\textstyle\sum_{p\in\mathcal P}\nabla_{W_p}R}_{g_q},
\end{equation}
where $\mathcal C$ and $\mathcal P$ index context and task positions and $W_p$ is the copy of $W$ used at $p$. This is the multi-task form of the state and prediction split of~\citep{monea2026state}, the context role playing the state role and the $K$ decisions playing $K$ prediction roles. Integrating \eqref{eq:grad-split} along the two paths, we model the risk on two formal copies of the block as additive,
\begin{equation}
\label{eq:additive}
R(\Theta_c,\Theta_q)=L_c(\Theta_c)+L_q(\Theta_q),
\end{equation}
with $\Theta_c$ used wherever the block acts in the context role and $\Theta_q$ wherever it acts in the task role. The tied architecture imposes $\Theta_c=\Theta_q$; the decoupled one does not.

\begin{proof}[Proof of Proposition~\ref{prop:source}]
The untied problem separates by \eqref{eq:additive} and attains $L_c(\Theta_c^\ast)+L_q(\Theta_q^\ast)$. The tied problem minimizes $J(\Theta)=L_c(\Theta)+L_q(\Theta)$; its stationarity condition $\nabla L_c(\Theta^t)+\nabla L_q(\Theta^t)=0$ gives
\begin{equation}
\label{eq:stationary}
\nabla L_c(\Theta^t)=-\nabla L_q(\Theta^t),
\end{equation}
so, unless both gradients vanish, which requires a common minimizer, the two roles pull with equal and opposite force and neither sits at its own optimum. This is the first claim.

For the magnitude, model each role loss near its minimizer by $L_r(\Theta)=L_r(\Theta_r^\ast)+\tfrac12(\Theta-\Theta_r^\ast)^\top H_r(\Theta-\Theta_r^\ast)$
for $r\in\{c,q\}$, with $H_c,H_q\succ0$. Minimizing the sum of two quadratics is a standard identity: for $A,B\succ0$,
\begin{equation}
\label{eq:quad-combine}
\min_{\Theta}\ \tfrac12(\Theta-a)^\top A(\Theta-a)+\tfrac12(\Theta-b)^\top B(\Theta-b)
=\tfrac12(a-b)^\top\big(A^{-1}+B^{-1}\big)^{-1}(a-b),
\end{equation}
attained at $\Theta=(A+B)^{-1}(Aa+Bb)$. Applying \eqref{eq:quad-combine} with $a=\Theta_c^\ast$, $b=\Theta_q^\ast$, $A=H_c$, $B=H_q$ gives the location $\Theta^t=(H_c+H_q)^{-1}(H_c\Theta_c^\ast+H_q\Theta_q^\ast)$ and the excess
\[
J(\Theta^t)-\big[L_c(\Theta_c^\ast)+L_q(\Theta_q^\ast)\big]
=\tfrac12\,\Delta^\top\big(H_c^{-1}+H_q^{-1}\big)^{-1}\Delta,
\qquad \Delta=\Theta_c^\ast-\Theta_q^\ast .
\]
The combined matrix $(H_c^{-1}+H_q^{-1})^{-1}=H_c(H_c+H_q)^{-1}H_q$ is symmetric positive definite, so the excess is non-negative and vanishes exactly when $\Delta=0$, which is \eqref{eq:price}.
\end{proof}

\subsection{Relation term}

\begin{proof}[Proof of Proposition~\ref{prop:relation}]

\emph{The gap on a fixed encoder.} By Lemma~\ref{lem:gibbs} with target $\mathbf y$
and $u=z$, the best joint read-out on $z$ attains $H(\mathbf y\mid z)$. By
Lemma~\ref{lem:rtae} the best factorized read-out attains $\sum_k H(y_k\mid z)$. The
difference is $\sum_k H(y_k\mid z)-H(\mathbf y\mid z)=\mathrm{TC}(\mathbf y\mid z)$,
which by the same divergence identity used in Step 5 of Theorem~\ref{thm:unified},
now conditioned on $z$, equals
$\mathbb E_z D_{\mathrm{KL}}\big(P(\mathbf y\mid z)\|\prod_k P(y_k\mid z)\big)\ge0$
and vanishes exactly under conditional independence given $z$.

\emph{The relation between the two total correlations.} Subtract the two definitions
and apply Step 3 of Theorem~\ref{thm:unified} once to each marginal and once to the
joint,
\[
\mathrm{TC}(\mathbf y\mid z)-\mathrm{TC}(\mathbf y\mid\mathcal A)
=\sum_k\big[H(y_k\mid z)-H(y_k\mid\mathcal A)\big]-\big[H(\mathbf y\mid z)-H(\mathbf y\mid\mathcal A)\big]
\]
which is \eqref{eq:tc-relation}. Neither side dominates in general, since the first
sum counts information that several decisions share once per decision while the
second counts it once, and both sides vanish together when $z$ is jointly sufficient.

\end{proof}

For any total order $\prec$ that extends the funnel order, the entropy chain rule
gives $H(\mathbf y\mid\mathcal A)=\sum_k H(y_k\mid y_{\prec k},\mathcal A)$, hence
\begin{equation}
\label{eq:tc-chain}
\mathrm{TC}(\mathbf y\mid\mathcal A)=\sum_k I\big(y_k;y_{\prec k}\mid\mathcal A\big),
\end{equation}
which we read as the total value of conditioning each decision on its predecessors.

\begin{proof}[Proof of Proposition~\ref{prop:escape}]
Assume that the parent assignment $\mathrm{pa}$ induces an acyclic graph and that
$\prec$ is one of its linear extensions, so $\mathrm{pa}(k)\subseteq\{y_j:j\prec k\}$
for every $k$ and $\prod_k q_k(y_k\mid y_{\mathrm{pa}(k)},\mathcal A)$ is a normalized
joint law. Acyclicity is required, since with a cyclic assignment the product is only
a pseudo-likelihood, Lemma~\ref{lem:gibbs} does not apply to it, and the resulting
excess can be negative.

\emph{Step 1, the optimum of the parent-conditioned family.} Apply
Lemma~\ref{lem:gibbs} to the $k$-th factor with target $y_k$ and conditioning
variable $u=(y_{\mathrm{pa}(k)},\mathcal A)$. As in Lemma~\ref{lem:rtae} the objective
is a sum whose $k$-th summand depends on $q_k$ only, so the infimum splits and equals
$\sum_k H(y_k\mid y_{\mathrm{pa}(k)},\mathcal A)$.

\emph{Step 2, subtracting the baseline term by term.} Write
$R^\star=\sum_k H(y_k\mid y_{\prec k},\mathcal A)$, which is the chain rule along $\prec$. Subtracting term by term,
\[
\sum_k\Big[H(y_k\mid y_{\mathrm{pa}(k)},\mathcal A)-H(y_k\mid y_{\prec k},\mathcal A)\Big]
=\sum_k I\big(y_k;\,y_{\prec k}\setminus\mathrm{pa}(k)\ \big|\ y_{\mathrm{pa}(k)},\mathcal A\big),
\]
where each equality uses $\mathrm{pa}(k)\subseteq\{y_j:j\prec k\}$, so that passing from the conditioning set $(y_{\mathrm{pa}(k)},\mathcal A)$ to the larger set $(y_{\prec k},\mathcal A)$ removes exactly the variables $y_{\prec k}\setminus\mathrm{pa}(k)$ and the entropy difference is by definition the conditional mutual information with those variables. Every term is non-negative, so
the total is non-negative.

\emph{Step 3, invariance to the extension.} The total equals
$\sum_k H(y_k\mid y_{\mathrm{pa}(k)},\mathcal A)-H(\mathbf y\mid\mathcal A)$. The first sum depends only on the fixed parent assignment and not on $\prec$, and the second is order free, so the total is the same for every linear extension of the same graph, even though the individual terms in Step 2 do depend on the extension.

\emph{Step 4, the two boundary cases.} Taking $\mathrm{pa}(k)=\{y_j:j\prec k\}$ makes $y_{\prec k}\setminus\mathrm{pa}(k)=\varnothing$, so every term is zero and the total vanishes. Taking every $\mathrm{pa}(k)=\varnothing$ turns the total into $\sum_k I(y_k;y_{\prec k}\mid\mathcal A)$, which is $\mathrm{TC}(\mathbf y\mid\mathcal A)$ by \eqref{eq:tc-chain}.
\end{proof}

\section{Computation Complexity}
\label{app:cc}

\begin{table}[ht]
\centering
\caption{Measured cost with the task head fixed to PLE ($L{=}16$, $d{=}96$, batch size $1$). Best in \textbf{bold}.}
\label{tab:cost}
\setlength{\tabcolsep}{3pt}
\begin{tabular}{l rr}
\toprule
Encoder & Params & FLOPs \\
\midrule
IntTravel          & \textbf{2.23M} & 3.04G \\
OneTrans           & 12.79M & 2.70G \\
HGenPush           & 3.58M & 3.38G \\
\textbf{\modelname (Ours)} & 2.98M & \textbf{1.72G} \\
\bottomrule
\end{tabular}
\end{table}

Let $N_s$ be the number of sessions, $N_c$ the number of context tokens, $N_q=K N_s$ the number of task tokens ($K$ tasks per session), $d$ the width, $L$ the depth.

A unified single-stream encoder that places task tokens in the same sequence attends over $N_c+N_q$ positions, costing $\mathcal O\!\big(L(N_c+N_q)^2 d\big)$ for attention. \modelname{}
instead performs three attentions per layer, i.e., context self-attention, task$\to$context cross-attention, and task self-attention. For a total of:
\begin{equation}
\mathcal O\!\big(L(N_c^2+N_qN_c+N_q^2)\,d\big),
\end{equation}
Compared with unified self-attention, this removes one cross-attention term, saving $\mathcal O(LN_cN_qd)$ operations. Finally, read-out differs asymptotically: a factorized head on a unified stream is evaluated at all $N_c+N_q$ positions, whereas in \modelname each task head consumes only its own $N_s$ task tokens, reducing head cost from $\mathcal O(KN d_{\text{head}})$ to $\mathcal O(KN_s d_{\text{head}})$. Table~\ref{tab:cost} reports measured cost with the head fixed to PLE. We count dense parameters of the encoder, read-out gate and head (excluding sparse embedding tables) and operator-level FLOPs.

%% file: main.bbl
\begin{thebibliography}{44}
\providecommand{\natexlab}[1]{#1}
\providecommand{\url}[1]{\texttt{#1}}
\expandafter\ifx\csname urlstyle\endcsname\relax
  \providecommand{\doi}[1]{doi: #1}\else
  \providecommand{\doi}{doi: \begingroup \urlstyle{rm}\Url}\fi

\bibitem[Bengio et~al.(2013)Bengio, Courville, and Vincent]{bengio2013representation}
Yoshua Bengio, Aaron Courville, and Pascal Vincent.
\newblock Representation learning: A review and new perspectives.
\newblock \emph{IEEE transactions on pattern analysis and machine intelligence}, 35\penalty0 (8):\penalty0 1798--1828, 2013.

\bibitem[Caruana(1997)]{caruana1997multitask}
Rich Caruana.
\newblock Multitask learning.
\newblock \emph{Machine learning}, 28\penalty0 (1):\penalty0 41--75, 1997.

\bibitem[Chen et~al.(2018)Chen, Badrinarayanan, Lee, and Rabinovich]{chen2018gradnorm}
Zhao Chen, Vijay Badrinarayanan, Chen-Yu Lee, and Andrew Rabinovich.
\newblock Gradnorm: Gradient normalization for adaptive loss balancing in deep multitask networks.
\newblock In \emph{International conference on machine learning}, pp.\  794--803. PMLR, 2018.

\bibitem[Cheng et~al.(2011)Cheng, Caverlee, Lee, and Sui]{cheng2011exploring}
Zhiyuan Cheng, James Caverlee, Kyumin Lee, and Daniel Sui.
\newblock Exploring millions of footprints in location sharing services.
\newblock In \emph{Proceedings of the International AAAI Conference on Web and Social Media}, pp.\  81--88, 2011.

\bibitem[Cho et~al.(2011)Cho, Myers, and Leskovec]{cho2011friendship}
Eunjoon Cho, Seth~A Myers, and Jure Leskovec.
\newblock Friendship and mobility: user movement in location-based social networks.
\newblock In \emph{Proceedings of the 17th ACM SIGKDD international conference on Knowledge discovery and data mining}, pp.\  1082--1090, 2011.

\bibitem[Cover(1999)]{cover1999elements}
Thomas~M Cover.
\newblock \emph{Elements of information theory}.
\newblock John Wiley \& Sons, 1999.

\bibitem[Deng et~al.(2025)Deng, Wang, Cai, Ren, Hu, Ding, Luo, and Zhou]{deng2025onerec}
Jiaxin Deng, Shiyao Wang, Kuo Cai, Lejian Ren, Qigen Hu, Weifeng Ding, Qiang Luo, and Guorui Zhou.
\newblock Onerec: Unifying retrieve and rank with generative recommender and iterative preference alignment.
\newblock \emph{arXiv preprint arXiv:2502.18965}, 2025.

\bibitem[Gneiting \& Raftery(2007)Gneiting and Raftery]{gneiting2007strictly}
Tilmann Gneiting and Adrian~E Raftery.
\newblock Strictly proper scoring rules, prediction, and estimation.
\newblock \emph{Journal of the American statistical Association}, 102\penalty0 (477):\penalty0 359--378, 2007.

\bibitem[Hou et~al.(2026)Hou, Wu, Liao, Yang, Zhang, Wang, Zheng, Wu, and Hong]{hou2026survey}
Min Hou, Le~Wu, Yuxin Liao, Yonghui Yang, Zhen Zhang, Yu~Wang, Changlong Zheng, Han Wu, and Richang Hong.
\newblock A survey on generative recommendation: Data, model, and tasks.
\newblock \emph{AI Open}, 2026.

\bibitem[Liang et~al.(2026)Liang, Feng, Feng, Wang, Ye, Feng, Deng, Zhang, Sun, Li, et~al.]{liang2026hgenpush}
Xiao Liang, Jiali Feng, Xin Feng, Yiqing Wang, Baolin Ye, Siyao Feng, Zhihui Deng, Cunyi Zhang, Huajin Sun, Xuanping Li, et~al.
\newblock Hgenpush: A heterogeneous generative recommendation architecture for industrial push notification systems.
\newblock \emph{arXiv preprint arXiv:2607.03362}, 2026.

\bibitem[Liu et~al.(2021)Liu, Liu, Jin, Stone, and Liu]{liu2021conflict}
Bo~Liu, Xingchao Liu, Xiaojie Jin, Peter Stone, and Qiang Liu.
\newblock Conflict-averse gradient descent for multi-task learning.
\newblock \emph{Advances in neural information processing systems}, 34:\penalty0 18878--18890, 2021.

\bibitem[Ma et~al.(2018{\natexlab{a}})Ma, Zhao, Yi, Chen, Hong, and Chi]{ma2018modeling}
Jiaqi Ma, Zhe Zhao, Xinyang Yi, Jilin Chen, Lichan Hong, and Ed~H Chi.
\newblock Modeling task relationships in multi-task learning with multi-gate mixture-of-experts.
\newblock In \emph{Proceedings of the 24th ACM SIGKDD international conference on knowledge discovery \& data mining}, pp.\  1930--1939, 2018{\natexlab{a}}.

\bibitem[Ma et~al.(2018{\natexlab{b}})Ma, Zhao, Huang, Wang, Hu, Zhu, and Gai]{ma2018entire}
Xiao Ma, Liqin Zhao, Guan Huang, Zhi Wang, Zelin Hu, Xiaoqiang Zhu, and Kun Gai.
\newblock Entire space multi-task model: An effective approach for estimating post-click conversion rate.
\newblock In \emph{The 41st International ACM SIGIR Conference on Research \& Development in Information Retrieval}, pp.\  1137--1140, 2018{\natexlab{b}}.

\bibitem[Misra et~al.(2016)Misra, Shrivastava, Gupta, and Hebert]{misra2016cross}
Ishan Misra, Abhinav Shrivastava, Abhinav Gupta, and Martial Hebert.
\newblock Cross-stitch networks for multi-task learning.
\newblock In \emph{Proceedings of the IEEE conference on computer vision and pattern recognition}, pp.\  3994--4003, 2016.

\bibitem[Monea et~al.(2026)Monea, Godey, Brantley, and Artzi]{monea2026state}
Giovanni Monea, Nathan Godey, Kiant{\'e} Brantley, and Yoav Artzi.
\newblock The state-prediction separation hypothesis.
\newblock \emph{arXiv preprint arXiv:2607.01218}, 2026.

\bibitem[Monti et~al.(2018)Monti, Palumbo, Rizzo, Troncy, Ehrhart, and Morisio]{monti2018semantic}
Diego Monti, Enrico Palumbo, Giuseppe Rizzo, Rapha{\"e}l Troncy, Thibault Ehrhart, and Maurizio Morisio.
\newblock Semantic trails of city explorations: How do we live a city.
\newblock \emph{arXiv preprint arXiv:1812.04367}, 2018.

\bibitem[Navon et~al.(2022)Navon, Shamsian, Achituve, Maron, Kawaguchi, Chechik, and Fetaya]{navon2022multi}
Aviv Navon, Aviv Shamsian, Idan Achituve, Haggai Maron, Kenji Kawaguchi, Gal Chechik, and Ethan Fetaya.
\newblock Multi-task learning as a bargaining game.
\newblock \emph{arXiv preprint arXiv:2202.01017}, 2022.

\bibitem[Oord et~al.(2018)Oord, Li, and Vinyals]{oord2018representation}
Aaron van~den Oord, Yazhe Li, and Oriol Vinyals.
\newblock Representation learning with contrastive predictive coding.
\newblock \emph{arXiv preprint arXiv:1807.03748}, 2018.

\bibitem[Rajput et~al.(2023)Rajput, Mehta, Singh, Keshavan, Vu, Heldt, Hong, Tay, Tran, Samost, et~al.]{rajput2023recommender}
Shashank Rajput, Nikhil Mehta, Anima Singh, Raghunandan~Hulikal Keshavan, Trung Vu, Lukasz Heldt, Lichan Hong, Yi~Tay, Vinh~Q Tran, Jonah Samost, et~al.
\newblock Recommender systems with generative retrieval.
\newblock In \emph{Thirty-seventh Conference on Neural Information Processing Systems}, 2023.

\bibitem[Sheng et~al.(2021)Sheng, Zhao, Zhou, Ding, Dai, Luo, Yang, Lv, Zhang, Deng, et~al.]{sheng2021one}
Xiang-Rong Sheng, Liqin Zhao, Guorui Zhou, Xinyao Ding, Binding Dai, Qiang Luo, Siran Yang, Jingshan Lv, Chi Zhang, Hongbo Deng, et~al.
\newblock One model to serve all: Star topology adaptive recommender for multi-domain ctr prediction.
\newblock In \emph{Proceedings of the 30th ACM International Conference on Information \& Knowledge Management}, pp.\  4104--4113, 2021.

\bibitem[Shi et~al.(2024)Shi, Si, Xu, Zhang, Zang, Zheng, Leng, Niu, and Song]{shi2024unisar}
Teng Shi, Zihua Si, Jun Xu, Xiao Zhang, Xiaoxue Zang, Kai Zheng, Dewei Leng, Yanan Niu, and Yang Song.
\newblock Unisar: Modeling user transition behaviors between search and recommendation.
\newblock In \emph{Proceedings of the 47th International ACM SIGIR Conference on Research and Development in Information Retrieval}, pp.\  1029--1039, 2024.

\bibitem[Shi et~al.(2026)Shi, Qin, Yu, Zhang, He, Fan, and Xu]{shi2026mining}
Teng Shi, Weicong Qin, Weijie Yu, Xiao Zhang, Ming He, Jianping Fan, and Jun Xu.
\newblock Mining informative interests via latent cross reasoning for search enhanced recommendation.
\newblock In \emph{Proceedings of the 49th International ACM SIGIR Conference on Research and Development in Information Retrieval}, pp.\  1612--1622, 2026.

\bibitem[Su et~al.(2024)Su, Pan, Wang, Xiao, Quan, Chen, and Jiang]{su2024stem}
Liangcai Su, Junwei Pan, Ximei Wang, Xi~Xiao, Shijie Quan, Xihua Chen, and Jie Jiang.
\newblock Stem: unleashing the power of embeddings for multi-task recommendation.
\newblock In \emph{Proceedings of the AAAI conference on artificial intelligence}, pp.\  9002--9010, 2024.

\bibitem[Sun et~al.(2020)Sun, Panda, Feris, and Saenko]{sun2020adashare}
Ximeng Sun, Rameswar Panda, Rogerio Feris, and Kate Saenko.
\newblock Adashare: Learning what to share for efficient deep multi-task learning.
\newblock \emph{Advances in Neural Information Processing Systems}, 33:\penalty0 8728--8740, 2020.

\bibitem[Tang et~al.(2020)Tang, Liu, Zhao, and Gong]{tang2020progressive}
Hongyan Tang, Junning Liu, Ming Zhao, and Xudong Gong.
\newblock Progressive layered extraction (ple): A novel multi-task learning (mtl) model for personalized recommendations.
\newblock In \emph{Proceedings of the 14th ACM conference on recommender systems}, pp.\  269--278, 2020.

\bibitem[Tang et~al.(2026)Tang, Dai, Wang, Guo, Zhao, Fu, Wu, Ni, Zeng, Chen, et~al.]{tang2026onerank}
Jiakai Tang, Sunhao Dai, Kun Wang, Zhiluohan Guo, Yu~Zhao, Cong Fu, Kangle Wu, Yabo Ni, Anxiang Zeng, Xu~Chen, et~al.
\newblock Onerank: Unified transformer-native ranking architecture for multi-task recommendation.
\newblock \emph{arXiv preprint arXiv:2606.16838}, 2026.

\bibitem[Vafidis et~al.(2025)Vafidis, Bhargava, and Rangel]{vafidis2025disentangling}
Pantelis Vafidis, Aman Bhargava, and Antonio Rangel.
\newblock Disentangling representations through multi-task learning.
\newblock In \emph{International Conference on Learning Representations}, volume 2025, pp.\  68296--68338, 2025.

\bibitem[Ver~Steeg \& Galstyan(2014)Ver~Steeg and Galstyan]{ver2014discovering}
Greg Ver~Steeg and Aram Galstyan.
\newblock Discovering structure in high-dimensional data through correlation explanation.
\newblock \emph{Advances in neural information processing systems}, 27, 2014.

\bibitem[Wang et~al.(2025)Wang, Cao, Fu, Gai, and Zhou]{wang2025home}
Xu~Wang, Jiangxia Cao, Zhiyi Fu, Kun Gai, and Guorui Zhou.
\newblock Home: Hierarchy of multi-gate experts for multi-task learning at kuaishou.
\newblock In \emph{Proceedings of the 31st ACM SIGKDD Conference on Knowledge Discovery and Data Mining V. 1}, pp.\  2638--2647, 2025.

\bibitem[Wang et~al.(2023)Wang, Lam, Wong, Liu, Zhao, Wang, Chen, Guo, and Tang]{wang2023multi}
Yuhao Wang, Ha~Tsz Lam, Yi~Wong, Ziru Liu, Xiangyu Zhao, Yichao Wang, Bo~Chen, Huifeng Guo, and Ruiming Tang.
\newblock Multi-task deep recommender systems: A survey.
\newblock \emph{arXiv preprint arXiv:2302.03525}, 2023.

\bibitem[Watanabe(1960)]{watanabe1960information}
Satosi Watanabe.
\newblock Information theoretical analysis of multivariate correlation.
\newblock \emph{IBM Journal of research and development}, 4\penalty0 (1):\penalty0 66--82, 1960.

\bibitem[Xi et~al.(2021)Xi, Chen, Yan, Zhang, Zhu, Zhuang, and Chen]{xi2021modeling}
Dongbo Xi, Zhen Chen, Peng Yan, Yinger Zhang, Yongchun Zhu, Fuzhen Zhuang, and Yu~Chen.
\newblock Modeling the sequential dependence among audience multi-step conversions with multi-task learning in targeted display advertising.
\newblock In \emph{Proceedings of the 27th ACM SIGKDD Conference on Knowledge Discovery \& Data Mining}, pp.\  3745--3755, 2021.

\bibitem[Yan et~al.(2022)Yan, Wang, Zhang, Li, Deng, Xu, and Zheng]{yan2022apg}
Bencheng Yan, Pengjie Wang, Kai Zhang, Feng Li, Hongbo Deng, Jian Xu, and Bo~Zheng.
\newblock Apg: Adaptive parameter generation network for click-through rate prediction.
\newblock \emph{Advances in Neural Information Processing Systems}, 35:\penalty0 24740--24752, 2022.

\bibitem[Yan et~al.(2025)Yan, Xu, Sun, Ou, Luo, Tan, Cheng, Liu, and Chu]{yan2025intsr}
Huimin Yan, Longfei Xu, Junjie Sun, Ni~Ou, Wei Luo, Xing Tan, Ran Cheng, Kaikui Liu, and Xiangxiang Chu.
\newblock Intsr: An integrated generative framework for search and recommendation.
\newblock \emph{arXiv preprint arXiv:2509.21179}, 2025.

\bibitem[Yan et~al.(2026)Yan, Xu, Sun, Liu, Luo, Liu, and Chu]{yan2026inttravel}
Huimin Yan, Longfei Xu, Junjie Sun, Zheng Liu, Wei Luo, Kaikui Liu, and Xiangxiang Chu.
\newblock Inttravel: A real-world dataset and generative framework for integrated multi-task travel recommendation.
\newblock \emph{arXiv preprint arXiv:2602.11664}, 2026.

\bibitem[Yang et~al.(2014)Yang, Zhang, Zheng, and Yu]{yang2014modeling}
Dingqi Yang, Daqing Zhang, Vincent~W Zheng, and Zhiyong Yu.
\newblock Modeling user activity preference by leveraging user spatial temporal characteristics in lbsns.
\newblock \emph{IEEE Transactions on Systems, Man, and Cybernetics: Systems}, 45\penalty0 (1):\penalty0 129--142, 2014.

\bibitem[Yang et~al.(2016)Yang, Zhang, and Qu]{yang2016participatory}
Dingqi Yang, Daqing Zhang, and Bingqing Qu.
\newblock Participatory cultural mapping based on collective behavior data in location-based social networks.
\newblock \emph{ACM Transactions on Intelligent Systems and Technology (TIST)}, 7\penalty0 (3):\penalty0 1--23, 2016.

\bibitem[Yang et~al.(2019)Yang, Qu, Yang, and Cudre-Mauroux]{yang2019revisiting}
Dingqi Yang, Bingqing Qu, Jie Yang, and Philippe Cudre-Mauroux.
\newblock Revisiting user mobility and social relationships in lbsns: a hypergraph embedding approach.
\newblock In \emph{The world wide web conference}, pp.\  2147--2157, 2019.

\bibitem[Yu et~al.(2024)Yu, Duan, Xu, Chen, Liu, Liu, Yang, Chu, and Guo]{yu2024dsfnet}
Jiahao Yu, Yihai Duan, Longfei Xu, Chao Chen, Shuliang Liu, Kaikui Liu, Fan Yang, Xiangxiang Chu, and Ning Guo.
\newblock Dsfnet: Learning disentangled scenario factorization for multi-scenario route ranking.
\newblock \emph{arXiv preprint arXiv:2404.00243}, 2024.

\bibitem[Yu et~al.(2020)Yu, Kumar, Gupta, Levine, Hausman, and Finn]{yu2020gradient}
Tianhe Yu, Saurabh Kumar, Abhishek Gupta, Sergey Levine, Karol Hausman, and Chelsea Finn.
\newblock Gradient surgery for multi-task learning.
\newblock \emph{Advances in neural information processing systems}, 33:\penalty0 5824--5836, 2020.

\bibitem[Zhai et~al.(2024)Zhai, Liao, Liu, Wang, Li, Cao, Gao, Gong, Gu, He, et~al.]{zhai2024actions}
Jiaqi Zhai, Lucy Liao, Xing Liu, Yueming Wang, Rui Li, Xuan Cao, Leon Gao, Zhaojie Gong, Fangda Gu, Michael He, et~al.
\newblock Actions speak louder than words: Trillion-parameter sequential transducers for generative recommendations.
\newblock \emph{arXiv preprint arXiv:2402.17152}, 2024.

\bibitem[Zhang et~al.(2025)Zhang, Yin, Yang, and Wang]{zhang2025advances}
Mingzhu Zhang, Ruiping Yin, Zhen Yang, and Yipeng Wang.
\newblock Advances and challenges of multi-task learning method in recommender systems: A survey.
\newblock \emph{Neurocomputing}, pp.\  132510, 2025.

\bibitem[Zhang et~al.(2022)Zhang, Liao, Liu, Xu, and Zheng]{zhang2022leaving}
Qianqian Zhang, Xinru Liao, Quan Liu, Jian Xu, and Bo~Zheng.
\newblock Leaving no one behind: A multi-scenario multi-task meta learning approach for advertiser modeling.
\newblock In \emph{Proceedings of the Fifteenth ACM International Conference on Web Search and Data Mining}, pp.\  1368--1376, 2022.

\bibitem[Zhang et~al.(2026)Zhang, Pei, Guo, Wang, Feng, Sun, Liu, and Sun]{zhang2026onetrans}
Zhaoqi Zhang, Haolei Pei, Jun Guo, Tianyu Wang, Yufei Feng, Hui Sun, Shaowei Liu, and Aixin Sun.
\newblock Onetrans: Unified feature interaction and sequence modeling with one transformer in industrial recommender.
\newblock In \emph{Proceedings of the ACM Web Conference 2026}, pp.\  8162--8170, 2026.

\end{thebibliography}
